\documentclass[runningheads]{llncs}
\usepackage{amsmath}
\usepackage{amssymb}
\usepackage{latexsym}
\usepackage[dvips]{graphicx}
\usepackage{psfrag}
\usepackage{times}
\usepackage[T1]{fontenc}
\usepackage[latin1]{inputenc}
\usepackage[margin=1.5in]{geometry}
\usepackage{graphicx}
\usepackage[linesnumbered,ruled,vlined]{algorithm2e}
\usepackage{color}
\usepackage[dvipsnames]{xcolor}

\def\qed{\rule{1.5mm}{3mm}}

\def\boxit#1{\vbox{\hrule\hbox{\vrule\kern4pt
  \vbox{\kern1pt#1\kern1pt}
\kern2pt\vrule}\hrule}}

\usepackage{xcolor}

\usepackage[linesnumbered,ruled,vlined]{algorithm2e}
\usepackage{pgf,tikz}
\usetikzlibrary{arrows,backgrounds,fit,positioning,chains,shapes,trees,automata,calc}

\usepackage{hyperref}

\begin{document}

\title{Subset Sum Problems With Digraph Constraints}
\author{Laurent Gourv\`es\inst{1} \and J\'er\^ome Monnot\inst{1} \and Lydia Tlilane\inst{1}}

 \institute{Universit\'e Paris-Dauphine, PSL Research University, CNRS, LAMSADE, 75016 Paris, France
 \\\email{$\{$laurent.gourves,jerome.monnot,lydia.tlilane$\}$@dauphine.fr}
}

\maketitle

\begin{abstract}
We introduce and study four optimization problems that generalize the well-known subset sum problem. Given a node-weighted digraph, select a subset of vertices whose total weight does not exceed a given budget. Some additional constraints need to be satisfied. The (weak resp.) digraph constraint imposes that if (all incoming nodes of resp.) a node $x$ belongs to the solution, then the latter comprises all its outgoing nodes (node $x$ itself resp.). The maximality
constraint ensures that a solution cannot be extended without violating the budget or the (weak) digraph constraint. We study the complexity of these problems and we present some approximation results according to the type of digraph given in input, e.g. directed acyclic graphs and oriented trees.

\vspace{5mm}\noindent {\bf Key words.} Subset Sum, Maximal problems, digraph constraints, complexity, directed acyclic graphs, oriented trees, PTAS.
\end{abstract}

\section{Introduction}
This paper deals with four optimization problems which generalize the well-known {\sc Subset Sum Problem}  (SS in short). Given a digraph $G=(V,A)$ such that each $x \in V$ has a nonnegative weight $w(x)$, we search for $S \subseteq V$ satisfying some constraints. As for SS we have a {\em budget constraint} which imposes  that  $w(S) \equiv \sum_{x \in S} w(s)$ does not exceed a given bound $B$. We depart from SS by considering the following constraints. The {\em digraph constraint} imposes to insert a node in $S$ if one of its incoming nodes in $G$ appears in $S$. A weaker form, called {\em weak digraph constraint}, imposes to insert a node in $S$ if {\em all} its its incoming nodes in $G$ appear in $S$. Our last constraint requires maximality with respect to the previous constraints. A set $S$ satisfies the {\em maximality constraint} if there is no $S'\supset S$ satisfying the budget and the  (weak resp.) digraph constraint.

\smallskip

Given a digraph $G=(V,A)$ and a budget $B$, the four problems studied in this article are the following. {\sc Subset Sum with digraph constraints} (SSG in short) is to find $S$ that maximizes $w(S)$ under the budget and the digraph constraints. In {\sc Subset sum with weak digraph constraints} (SSGW in short), we seek $S$ that maximizes $w(S)$ under the budget and  the weak digraph constraints. For the
{\sc Maximal Subset Sum with digraph constraints} ({\sc Maximal} SSG in short) we search for  $S$ with minimum weight under the constraints of budget,  digraph and maximality. Finally,
{\sc Maximal Subset Sum with  weak digraph constraints} ({\sc Maximal} SSGW in short) aims to find $S$ with minimum weight under the constraints of budget, weak digraph and maximality.
The fact that we minimize $w(S)$ will become clear from the possible applications.

\smallskip

Let us motivate SSG in a scheduling context (other applications are given in \cite{kellerer2004knapsack}).
A processor is available during a period of length $B$
and there is a set of tasks to be executed on it. Each task $x$ is represented by a vertex of a digraph and has a duration $w(x)$. We seek a subset of tasks whose total duration does not exceed $B$.
The digraph provides dependency constraints between the tasks, i.e. there is an arc $(i,j)$ if task $i$ requires the output of task $j$. The goal is to maximize the utilization of the processor during the time window.

\smallskip

Keep this scheduling example but replace the processor by a lazy bureaucrat who has to execute some tasks. Everyday the bureaucrat is in his office for a period of length $B$ and the set of tasks $S$ that he selects must be executed within this period. The maximization of $w(S)$ does not reflect the wish of the lazy bureaucrat who is interested in working as little as possible. His goal is rather to minimize $w(S)$. However, $S=\emptyset$ is not a realistic solution because the worker's employer finds unacceptable to ignore a task if there is enough time to
execute it. This example, taken from \cite{ArkinBMS03}, motivates {\sc Maximal} SSG and its constraint of maximality.

\smallskip

Let us motivate  SSGW with another application. The different modules of a program are represented by a
digraph $G=(V,A)$ in the sense that $(x,y)$ belongs to $A$ whenever module $y$ receives informations from module $x$. An updated version of the program is to be deployed. Updating module $x$ induces
a cost of $w(x)$ and there is a global budget of $B$. We want to select a subset $S$ of modules, candidates for the update, such that $w(S) \leq B$ and if all the predecessors of a module $y$ are updated, then $y$ must also be updated otherwise $y$ works in a faulty manner.

\smallskip

In order to justify the study of {\sc Maximal}  SSGW, suppose the user of the program pays an external company
$B$ dollars for updating the software. If $S$ denotes the set of updated modules, then the revenue of the company, to be maximized, is equal to $B-w(S)$. Meanwhile, the user finds $S$ acceptable if $B$ is exceeded with any extra update.

%

\smallskip

Our purpose is to study SSG, SSGW, {\sc Maximal} SSG and {\sc Maximal} SSGW from a theoretical viewpoint. The complexity and approximability of these problems are analyzed for various topologies of the input digraph.

\smallskip

To the best of our knowledge, these problems are new, except SSG which is a special case of the \textsc{Partially-Ordered Knapsack} problem (also known as the \textsc{Precedence-Constrained Knapsack Problem} and it will define later) \cite{johnson1983knapsacks,kellerer2004knapsack,HajiaghayiJLMRV06}.
{\sc Maximal} SSG and {\sc Maximal} SSGW generalize the {\sc Lazy bureaucrat problem} with common deadlines and arrivals \cite{ArkinBMS03,EsfahbodGS03,GZ08,GMP13} representing the {\em maximal} version of SS.
We aim at extending this problem with (weak) digraph constraints on digraphs.
\smallskip

Our main results are: the four problems are {\bf NP}-hard, even for simple classes of digraphs. This is true even in in-rooted and out-rooted trees. SSG is also strongly {\bf NP}-hard in 3-regular digraphs and this result is tight according to degree parameters. {\sc Maximal} SSG is strongly {\bf NP}-hard in directed acyclic graphs (DAG) with a unique sink  and 3 weights. There is also a reduction preserving approximation for the four problems in DAG with maximum in-degree 2. However, some classes of graphs make the problems solvable in polynomial time or approximable within any given error.
The class of oriented trees admits non-trivial dynamic programming algorithms. In
tournament graphs, SSG and {\sc Maximal} SSG are polynomial. We also provide approximation schemes for SSG and {\sc Maximal} SSG
in DAG.
\smallskip

The present paper is organized as follows. Section  \ref{def} contains some definitions on graphs that we use throughout the paper and a formal definition of our problems. Section \ref{rel} makes an overview of related works. Then, we present some complexity results for the four problems according to the topology of the digraph: regular digraphs are studied in Section \ref{gen}, DAG in Section \ref{gen_prop} and oriented trees in Section \ref{tree}. Dynamic programming algorithms are provided for oriented trees in Section \ref{tree}. In Section \ref{sec:approx}, we propose approximation schemes for SSG and {\sc Maximal} SSG in DAG. Some perspectives are given in Section \ref{discuss}.

\section{Definitions and concepts}\label{def}

\subsection{Graph terminology}\label{term}
A directed graph (or {\em digraph}) is a graph whose edges have a direction. Formally, a digraph $G$ is a pair $(V,A)$ where $V$ and $A$ are the {\em vertex set} and the {\em arc set}, respectively.
Given two vertices $x$ and $y$, the notation $(x,y)$ means the arc that goes from $x$ to $y$ and $[x,y]$ is an edge (a non-oriented arc).
\smallskip

The {\em in-neighborhood} (and the {\em in-degree}) of a vertex $v$ in $G$ denoted by  $N^-_G(v)$ and $\deg^-_G(v)$ respectively are defined by $N^-_G(v)=\{u\in V:(u,v)\in A\}$ and $\deg^-_G(v)=|N^-_G(v)|$. Similarly, the {\em out-neighborhood} and the {\em out-degree}, $N^+_G(v)$ and $\deg^+_G(v)$ are
defined by $N^+_G(v)=\{u\in V:(v,u)\in A\}$ and $\deg^+_G(v)=|N^+_G(v)|$. A vertex with $\deg^-_G(v)=0$  is called a {\em source} and similarly, a vertex with $\deg^+_G(v)=0$ is called a {\em sink}. The neighborhoods of a vertex $v$ is defined by the set $N_G(v)=N^-_G(v) \cup N^+_G(v)$ and its degree is $\deg_G(v)=\deg^-_G(v)+\deg^+_G(v)$.
A graph is  $k$-regular if the degree of each node is $k$.
\smallskip

A {\em directed path} $\mu_G(x,y)$ from $x$ to $y$ in $G$ is a succession of vertices $(v_1,\dots,v_k)$ where $v_1=x$, $v_k=y$ and $(v_i,v_{i+1})\in A$ for every $i=1,\dots,k-1$. A {\em circuit} $C$ is a path of positive length from $x$ to $x$.
\smallskip

Given a subset of vertices $S\subseteq V$, we denote by $G-S$ the subgraph induced by $V\setminus S$.
\smallskip

In this paper, we also consider some special classes of digraphs: an {\em acyclic digraph} (or DAG for Directed Acyclic Graph) is a digraph without circuit. It is well known that a DAG  has a source and a sink. An {\em oriented tree} is a digraph formed by orienting the edges of an undirected tree and an {\em out-rooted tree} ({\em in-rooted tree} resp.) is an oriented tree where the out-degree (in-degree resp.) of each vertex is equal to 1. A root (anti-root resp.) is a vertex without any in-neighborhood (out-neighborhood resp.).
Finally, a {\em tournament} is an oriented graph where the underlying graph is a complete graph, or equivalently there is exactly one arc between  any two distinct vertices.
\smallskip

In this document, we only consider simple digraphs, i.e. with no loop and no multiple arc. 

\subsection{Subset Sum problems}\label{pbs}


\subsubsection{\textsc{Subset Sum with digraph constraints}} \

The first problem is called  {\sc subset sum with digraph constraints} (\textsc{SSG} in short) and its input is
a digraph $G=(V,A)$, a non-negative weight $w(i)$ for every node $i \in V$, and a positive bound $B$. The weight of $S \subseteq V$ is denoted by $w(S)$ and defined as $\sum_{i \in S} w(i)$. A feasible solution $S$ is a subset of $V$ satisfying the following constraints.
\begin{eqnarray}
&& \label{dc} \forall x \in S, \, (x,y) \in A \Rightarrow y \in S\\
&&\label{bd} w(S)\le B
\end{eqnarray}
Constraints (\ref{dc}) are called the {\em digraph constraints} while  (\ref{bd}) corresponds to a {\em budget constraint}. Formally, the problem is defined by:

\begin{center}
\begin{tabular}{|l|}
\hline
\textsc{SSG}\\
\hline
\textsl{Input}: a node weighted digraph $G=(V,A,w)$ and a bound $B$. \\
\textsl{Output}:  $S\subseteq V$ satisfying  (\ref{dc}) and (\ref{bd}).\\
\textsl{Objective}: maximize $w(S)$.\\
\hline
\end{tabular}
\end{center}

Obviously, this optimization problem is related to the exact decision version where we try to decide if there is a  subset
$S$ satisfying (\ref{dc}) with $w(S)= B$, which is a natural generalization of the standard {\sc Subset Sum} decision problem
(see Problem [SP13], page 223 in \cite{GJ79}) by considering the digraph without arcs, i.e., $A=\emptyset$.
\smallskip

Let us introduce an intermediary decision problem, called \textsc{Cardinality SSG} in the rest of the paper.
The input consists of a digraph $G=(V,A)$, a bound $B$, an integer $p \le |V|$ and a weight function $w$ on the nodes satisfying $1 \le w(v) \le B-p$. The problem is to decide if there exists $J \subseteq V$ such that $w(J)=B$, $|J|=p$ and $J$ satisfies the digraph constraints \eqref{dc}.

\subsubsection{\textsc{Maximal Subset Sum with digraph constraints}} \

This new problem is called  {\sc Maximal subset sum with digraph constraints} (\textsc{Maximal SSG} in short) and its input is the same as for
\textsc{SSG}. A feasible solution $S$ is a subset of $V$ satisfying $(\ref{dc})$, $(\ref{bd})$ and the following third constraint:
\begin{equation}
\label{busy1} \text{ There is no } S'\supset S \text{ such that } S' \text{satisfies \eqref{dc} and \eqref{bd}}.
\end{equation}

This last condition is called the {\em maximality constraint} and it corresponds to the notion of maximal subset satisfying the digraph constraint. As opposed to \textsc{SSG},
the goal of \textsc{Maximal SSG} is to {\em minimize} $w(S)$. Formally:

\begin{center}
\begin{tabular}{|l|}
\hline
\textsc{Maximal SSG}\\
\hline
\textsl{Input}: a node weighted digraph $G=(V,A,w)$ and a bound $B$. \\
\textsl{Output}:  $S\subseteq V$ satisfying  (\ref{dc}), (\ref{bd}) and (\ref{busy1}).\\
\textsl{Objective}: minimize $w(S)$.\\
\hline
\end{tabular}
\end{center}

\bigskip

\noindent We also strengthen the digraph constraints by a new kind of constraints called {\em weak digraph constraints} and defined by:
\begin{eqnarray}
&& \label{sdc}  N^-_G(x)\subseteq S \wedge N^-_G(x)\neq \emptyset \Rightarrow x\in S
\end{eqnarray}

By replacing (\ref{dc}) by (\ref{sdc}) in the definition of {\sc SSG} and {\sc Maximal SSG}, we obtain two additional optimization problems.


\subsubsection{\textsc{Subset Sum with weak digraph constraints}} \

This problem is called  {\sc subset sum with weak digraph constraints} (\textsc{SSGW} in short).

\begin{center}
\begin{tabular}{|l|}
\hline
\textsc{SSGW}\\
\hline
\textsl{Input}: a node weighted digraph $G=(V,A,w)$ and a bound $B$. \\
\textsl{Output}:  $S\subseteq V$ satisfying  (\ref{sdc}) and (\ref{bd}).\\
\textsl{Objective}: maximize $w(S)$.\\
\hline
\end{tabular}
\end{center}

\subsubsection{\textsc{Maximal Subset Sum with weak digraph constraints}} \

As previously, we can define a maximal subset satisfying the weak digraph constraint:

\begin{equation}
\label{sbusy1} \text{ There is no } S'\supset S \text{ such that } S' \text{satisfies \eqref{sdc} and \eqref{bd}}.
\end{equation}

This  condition is denoted by the {\em weak maximality constraint} and  the last  problem 
is called  {\sc Maximal subset sum with weak digraph constraints} (\textsc{Maximal SSGW} in short)
\begin{center}
\begin{tabular}{|l|}
\hline
\textsc{Maximal SSGW}\\
\hline
\textsl{Input}: a node weighted digraph $G=(V,A,w)$ and a bound $B$. \\
\textsl{Output}:  $S\subseteq V$ satisfying  (\ref{sdc}), (\ref{bd}) and (\ref{sbusy1}).\\
\textsl{Objective}: minimize $w(S)$.\\
\hline
\end{tabular}
\end{center}

The feasibility of a solution $S\subseteq V$ for \textsc{Maximal SSGW} can be decided in $O(n^2)$. Indeed, \eqref{sdc} and \eqref{bd} are checked in $O(n)$ and the maximality constraint \eqref{sbusy1} can be checked as follows: for every $v\in V\setminus S$, add $v$ and (inductively) the vertices that allow to satisfy  \eqref{sdc} in $S$ (by necessary condition) and check condition \eqref{bd} with the increased set because the weights are non-negative.

\section{Related works}\label{rel}

The \textsc{Subset Sum Problem} is one of the simplest and fundamental \textbf{NP}-hard problems. It appears in
many real world applications. Given $n$ integers $a_i$ for $i=1,\dots,n$ and a target $B$, the goal is to find a subset $S\subseteq \{1,\dots,n\}$ such that $\sum_{i\in S}a_i=B$. A survey of existing results on the \textsc{Subset Sum Problem} can be found in Chapter 4 of \cite{kellerer2004knapsack}. There are several generalizations of the  \textsc{Subset Sum Problem} studied in the literature, see for instance \cite{WoegingerY92,KothariSZ05,CieliebakEPS08,EggermontW13,bervoets2015}. In
\cite{WoegingerY92}, the variation, called \textsc{Equal Subset Sum from two sets} is shown to be \textbf{NP}-complete, where given a set of $n$ integers $a_i$ for $i=1,\dots,n$, the problem is to decide whether there exist two disjoint nonempty subsets of indices $S_1,S_2\subseteq \{1,\dots,n\}$ such that $\sum_{i\in S_1}a_i=\sum_{j\in S_2}a_j$. In \cite{KothariSZ05}, two generalizations to intervals are proposed and they are motivated by single-item multi-unit auctions; here, we are given a set of $n$ intervals $[a_i,b_i]$, a target $B$, and the goal is to choose a set of integers (at most one from each interval for the first problem and for the second problem, the additional restriction that  at least $k_1$ and at most $k_2$ integers must be selected), whose sum approximates $B$ as best as possible. Several results are proposed, including a \textbf{FPTAS}. In \cite{EggermontW13}, the problem of deciding whether all integer values between two given bounds $B^-$ and $B^+$ are attainable  is proved to be $\mathbf{\Pi^2_p}$-
complete.
\smallskip

Many variations of the \textsc{Subset Sum Problem} have also been studied in \cite{CieliebakEP03,CieliebakEPS08} and especially a version on undirected graphs called \textsc{ESS with Exclusions}. Given  a connected undirected exclusion graph $G=(V,E)$ where the nodes are weighted by $w(v)\geq 0$, the problem consists in deciding if there are two disjoint independent sets $X,Y\subseteq V$ of $G$ such that $w(X)=\sum_{x\in X}w(x)=\sum_{y\in Y}w(y)$. \textsc{ESS with Exclusions} is obviously \textbf{NP}-complete and a pseudo-polynomial time algorithm is presented in \cite{CieliebakEP03,CieliebakEPS08}.
\smallskip

The \textsc{Partially-Ordered Knapsack} problem (also known as the \textsc{Precedence - Constrained Knapsack Problem}) is a natural generalization of \textsc{SSG} (exactly as {\sc Knapsack} generalizes {\sc Subset Sum}).
Here, we are given a set $V$ of items, a DAG $G=(V,A)$ (or equivalently a poset $\prec_P$ on $V)$  and a bound $B$. Each item $v\in V$ has a size $p(v)\geq 0$ and an associated weight $w(v)\geq 0$. The objective is to find $S\subseteq V$, whose weight $w(S)=\sum_{v \in S}w(v)$ is maximized under the digraph constraints and also $p(S):=\sum_{v \in S} p(v)$ must be at most $B$.
When $p(v)=w(v)$, we clearly obtain \textsc{SSG}. \textsc{Partially-Ordered Knapsack} is strongly \textbf{NP}-hard, even when $p(v) = w(v)$, $\forall v\in V$, and $G$ is a bipartite DAG \cite{johnson1983knapsacks}. In 2006, it was demonstrated in \cite{HajiaghayiJLMRV06} that \textsc{Partially-Ordered Knapsack} is hard to approximate within a factor $2^{(\log n)^\varepsilon}$, for some $\varepsilon > 0$, unless \textsc{3SAT} $\in$ \textbf{DTIME}$(2^{n^{\frac{3}{4}+\delta}})$. A survey of some applications and results can also be found in the book (pages 402-408 of \cite{kellerer2004knapsack}).
\smallskip

In \cite{kolliopoulos2007},  some \textbf{FPTAS} are proposed for \textsc{Partially-Ordered Knapsack} in the case of $2$-dimensional partial ordering (a generalization of series-parallel digraphs) and in the DAG whose bipartite complement are chordal bipartite. Also, a polynomial-time algorithm for \textsc{Partially-Ordered Knapsack} on Red-Blue bipartite DAG is given when its bipartite complement is chordal bipartite.
\smallskip

In the case of rooted trees, a \textbf{FPTAS} is also proposed for \textsc{Precedence-Constrained Knapsack Problem} in
\cite{johnson1983knapsacks}. The special case of in-rooted trees is also known in the literature as the {\em Tree Knapsack Problem} \cite{BeckerP95,ChoS97}.
\smallskip

In \cite{BolandBFFS12} an approach based on clique inequalities is presented for determining facets of the polyhedron of the \textsc{Precedence-Constrained Knapsack Problem}.
\smallskip

Two related problems, known as the \textsc{Neighbour Knapsack Problem}, are studied in \cite{BorradaileHW12} in which dependencies between items are given by an undirected (or a directed) graph $G=(V,E)$. In the first version, an item can
be selected only if at least one of its neighbors is also selected. In the second version, an item can
be selected only when all its neighbors are also selected. The authors of \cite{BorradaileHW12} propose upper and lower bounds on the
approximation ratios for these two problems on undirected and directed graphs.
\smallskip

Concerning the {\em Maximal version} of {\sc Subset Sum}, this problem is called the \textsc{Lazy Bureaucrat Problem} with common arrivals and deadlines  in the literature \cite{GZ08,GMP13} where the problem has been proved \textbf{NP}-hard and approximable with a \textbf{FPTAS}. This latter problem has also several generalizations known as the \textsc{Lazy Bureaucrat scheduling problem} \cite{ArkinBMS03,EsfahbodGS03,GZ08} and the \textsc{Lazy Matroid Problem}  \cite{GourvesMP14}.

\section{Regular Digraphs}\label{gen}


\begin{theorem}\label{Strong-NPC}
\textsc{SSG} is strongly \textbf{NP}-hard for connected digraphs in which each node has either out-degree 2 and in-degree 1 or the reverse.
\end{theorem}
\begin{proof}
We prove the strong \textbf{NP}-hardness using a reduction from {\sc Clique}:
\smallskip

\begin{center}
\begin{tabular}{|l|}
\hline
\textsc{Clique}\\
\hline
\textsl{Input}: a connected simple graph $G=(V,E)$. \\
\textsl{Output}:  $V'\subseteq V$  such that every two vertices in $V'$ are joined by an edge in $E$.\\
\textsl{Objective}: maximize $|V'|$.\\
\hline
\end{tabular}
\end{center}

{\sc Clique} is known to be \textbf{NP}-hard, even in regular connected graphs of degree $\Delta\geq 3$ (Problem [GT19], page 194 in \cite{GJ79}).

Let $I=(G,k)$ be an instance of the decision version of {\sc Clique} where $G=(V,E)$ is a regular connected graph of degree $\Delta$ and $V=\{1,\dots,n\}$. We construct an instance $I'=(G'=(V',A'),w,B)$ of  \textsc{SSG} as follows:\\

$G'=(V',A')$ is a connected digraph defined by $V'=V_C\cup V_E$ where $V_C=\{v_{i,j}:i\in V, j\in N_G(i)\}$ and $V_E=\{v_e^i:e\in E,~i=1,\dots,6\}=\cup_{e\in E} H(e)$ where $H(e)=\{v_e^i,i=1,...,6\}$ is a gadget. We start from $G$, and we replace each node $i\in V$ by a circuit $C_i$ of $\Delta$ vertices $v_{i}^j$ for $j\in N_G(i)$. Then two circuits $C_i$ and $C_j$ are connected via a gadget $H(e)$ if edge $e=[i,j]\in E$, 
in such a way that each node of the circuit has in-degree 2 and out-degree 1. Formally, if $e=[x,y]\in E$, then the gadget $H(e)$ has 6 nodes $\{v_e^i:~i=1,\dots,6\}$ where $v_e^1=v_e^x$ and $v_e^6=v_e^y$. An illustration is given in  Figure \ref{GadHfig1}.

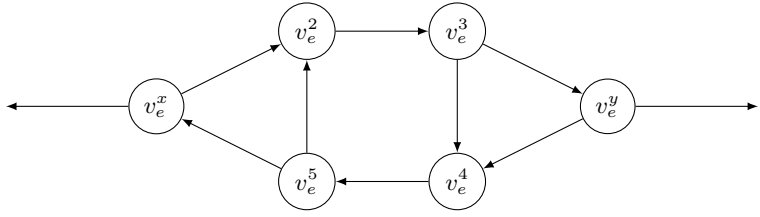
\begin{figure}[h!]
\centering

\begin{tikzpicture}
 \begin{scope}[every node/.style={draw,circle,black}]
\node(1) at (0,2) {$v_e^2$};
\node(2) at (2,2) {$v_e^3$} ;
\node(3) at (2,0) {$v_e^4$} ;
\node(4) at (0,0) {$v_e^5$} ;
\node(5) at (4,1) {$v_e^y$} ;
\node(6) at (-2,1) {$v_e^x$} ;
\end{scope}

\draw[-latex] (1) to (2) ;
\draw[-latex] (6) to (1) ;
\draw[-latex] (4) to (6) ;
\draw[-latex] (2) to (3) ;
\draw[-latex] (3) to (4) ;
\draw[-latex] (5) to (3) ;
\draw[-latex] (4) to (1) ;
\draw[-latex] (2) to (5) ;

\draw[-latex] (6) to (-4,1) ;
\draw[-latex] (5) to (6,1) ;

\end{tikzpicture}
 \caption{The gadget $H(e)$ for $e=[x,y]\in E$.}
 \label{GadHfig1}
\end{figure}

If $e=[i,j]\in E$, then we add the two arcs $(v_e^i,v_{i,j})$ and $(v_e^j,v_{j,i})$ in $G'$. Finally, each node of each circuit $C_i$ has weight $1$, i.e., $w(v_{i,j})=1$ while $w(v_e^i)=\Delta n$ for $i=1,\dots,6$. The bound $B=3\Delta nk(k-1)+\Delta k$.

An illustration of the construction is given in Figure \ref{fig2} for the graph described in  Figure \ref{fig22}. Clearly, this construction is done in polynomial time and  $G$ is a 3-regular connected digraph. Moreover, for each $v\in V'$ either $d^+_{G'}(v)=2$ and  $d^-_{G'}(v)=1$ , or $d^-_{G'}(v)=2$ and  $d^+_{G'}(v)=1$.

\smallskip

\begin{figure}[h!]
\centering

%
%

\begin{tikzpicture}
 \begin{scope}[every node/.style={draw,circle,black}]
\node(1) at (0,0) {$1$};
\node(2) at (2,0) {$2$} ;
\node(3) at (4,0) {$3$} ;
\node(4) at (6,0) {$4$} ;
\node(5) at (6,-2) {$5$} ;
\node(6) at (4,-2) {$6$} ;
\node(7) at (2,-2) {$7$} ;
\node(8) at (0,-2) {$8$} ;
\end{scope}

\draw (1) to (2) ;
\draw (1) to [bend left] (4) ;
\draw (1) to (7) ;
\draw (1) to (8) ;
\draw (2) to (3) ;
\draw (2) to (7) ;
\draw (2) to (8) ;
\draw (3) to (4) ;
\draw (3) to (5) ;
\draw (3) to (6) ;
\draw (4) to (5) ;
\draw (4) to (6) ;
\draw (5) to (6) ;
\draw (5) to [bend left] (8) ;
\draw (6) to (7) ;
\draw (7) to (8) ;

\end{tikzpicture}
 \caption{Example of an instance $G$ of {\sc Clique} with $\Delta=4$.}
 \label{fig22}
\end{figure}
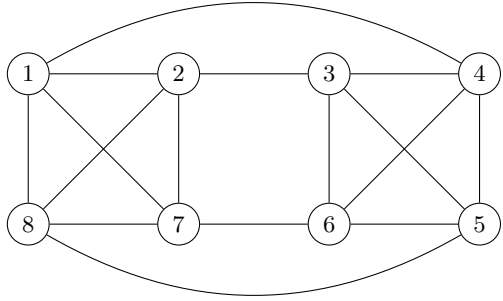

\begin{figure}[h!]
\centering

\begin{tikzpicture}[scale=1]

\begin{scope}[every node/.style={draw,circle,black,inner sep=0.8}]
\node(14) at (0,0) {$v_{1,4}$};
\node(12) at (1,0) {$v_{1,2}$};
\node(17) at (1,-1) {$v_{1,7}$};
\node(18) at (0,-1) {$v_{1,8}$};

\node(81) at (0,-3) {$v_{8,1}$};
\node(82) at (1,-3) {$v_{8,2}$};
\node(87) at (1,-4) {$v_{8,7}$};
\node(85) at (0,-4) {$v_{8,5}$};

\node(21) at (4,0) {$v_{2,1}$};
\node(23) at (5,0) {$v_{2,3}$};
\node(27) at (5,-1) {$v_{2,7}$};
\node(28) at (4,-1) {$v_{2,8}$};

\node(71) at (4,-3) {$v_{7,1}$};
\node(72) at (5,-3) {$v_{7,2}$};
\node(76) at (5,-4) {$v_{7,6}$};
\node(78) at (4,-4) {$v_{7,8}$};

\node(32) at (8,0) {$v_{3,2}$};
\node(34) at (9,0) {$v_{3,4}$};
\node(35) at (9,-1) {$v_{3,5}$};
\node(36) at (8,-1) {$v_{3,6}$};

\node(63) at (8,-3) {$v_{6,3}$};
\node(64) at (9,-3) {$v_{6,4}$};
\node(65) at (9,-4) {$v_{6,5}$};
\node(67) at (8,-4) {$v_{6,7}$};

\node(43) at (12,0) {$v_{4,3}$};
\node(41) at (13,0) {$v_{4,1}$};
\node(45) at (13,-1) {$v_{4,5}$};
\node(46) at (12,-1) {$v_{4,6}$};

\node(53) at (12,-3) {$v_{5,3}$};
\node(54) at (13,-3) {$v_{5,4}$};
\node(58) at (13,-4) {$v_{5,8}$};
\node(56) at (12,-4) {$v_{5,6}$};
\end{scope}

\node(h18) at (0,-2) {$H([1,8])$};
\node(h12) at (2.5,0) {$H([1,2])$};
\node(h27) at (5,-2) {$H([2,7])$};
\node(h78) at (2.5,-4) {$H([7,8])$};
\node(h17) at (3,-1.7) {$H([1,7])$};
\node(h28) at (2,-2.2) {$H([2,8])$};
\node(h23) at (6.5,0) {$H([2,3])$};
\node(h67) at (6.5,-4) {$H([6,7])$};
\node(h34) at (10.5,0) {$H([3,4])$};
\node(h56) at (10.5,-4) {$H([5,6])$};
\node(h36) at (8,-2) {$H([3,6])$};
\node(h45) at (13,-2) {$H([4,5])$};
\node(h35) at (11,-1.7) {$H([3,5])$};
\node(h46) at (10,-2.2) {$H([4,6])$};
\node(h14) at (6.5,1) {$H([1,4])$};
\node(h58) at (6.5,-5) {$H([5,8])$};

\draw[-latex] (14) to (12) ;
\draw[-latex] (12) to (17) ;
\draw[-latex] (17) to (18) ;
\draw[-latex] (18) to (14) ;

\draw[-latex] (81) to (82) ;
\draw[-latex] (82) to (87) ;
\draw[-latex] (87) to (85) ;
\draw[-latex] (85) to (81) ;

\draw[-latex] (21) to (23) ;
\draw[-latex] (23) to (27) ;
\draw[-latex] (27) to (28) ;
\draw[-latex] (28) to (21) ;

\draw[-latex] (71) to (72) ;
\draw[-latex] (72) to (76) ;
\draw[-latex] (76) to (78) ;
\draw[-latex] (78) to (71) ;

\draw[-latex] (32) to (34) ;
\draw[-latex] (34) to (35) ;
\draw[-latex] (35) to (36) ;
\draw[-latex] (36) to (32) ;

\draw[-latex] (43) to (41) ;
\draw[-latex] (41) to (45) ;
\draw[-latex] (45) to (46) ;
\draw[-latex] (46) to (43) ;

\draw[-latex] (53) to (54) ;
\draw[-latex] (54) to (58) ;
\draw[-latex] (58) to (56) ;
\draw[-latex] (56) to (53) ;

\draw[-latex] (63) to (64) ;
\draw[-latex] (64) to (65) ;
\draw[-latex] (65) to (67) ;
\draw[-latex] (67) to (63) ;

\draw[-latex] (h18) to (18) ;
\draw[-latex] (h18) to (81) ;
\draw[-latex] (h12) to (12) ;
\draw[-latex] (h12) to (21) ;
\draw[-latex] (h27) to (27) ;
\draw[-latex] (h27) to (72) ;
\draw[-latex] (h78) to (78) ;
\draw[-latex] (h78) to (87) ;
\draw[-latex] (h17) to [bend right] (17) ;
\draw[-latex] (h17) to [bend right] (71) ;
\draw[-latex] (h28) to [bend left] (28) ;
\draw[-latex] (h28) to [bend left] (82) ;
\draw[-latex] (h23) to (23) ;
\draw[-latex] (h23) to (32) ;
\draw[-latex] (h67) to (76) ;
\draw[-latex] (h67) to (67) ;
\draw[-latex] (h34) to (34) ;
\draw[-latex] (h34) to (43) ;
\draw[-latex] (h56) to (56) ;
\draw[-latex] (h56) to (65) ;
\draw[-latex] (h36) to (36) ;
\draw[-latex] (h36) to (63) ;
\draw[-latex] (h45) to (45) ;
\draw[-latex] (h45) to (54) ;
\draw[-latex] (h35) to [bend right] (35) ;
\draw[-latex] (h35) to [bend right] (53) ;
\draw[-latex] (h46) to [bend left] (46) ;
\draw[-latex] (h46) to [bend left] (64) ;
\draw[-latex] (h14) to [bend right] (14) ;
\draw[-latex] (h14) to [bend left] (41) ;
\draw[-latex] (h58) to [bend right] (58) ;
\draw[-latex] (h58) to [bend left] (85) ;

\end{tikzpicture}
 \caption{Example of digraph $G'$ constructed in the reduction from $G$.}
 \label{fig2}
\end{figure}
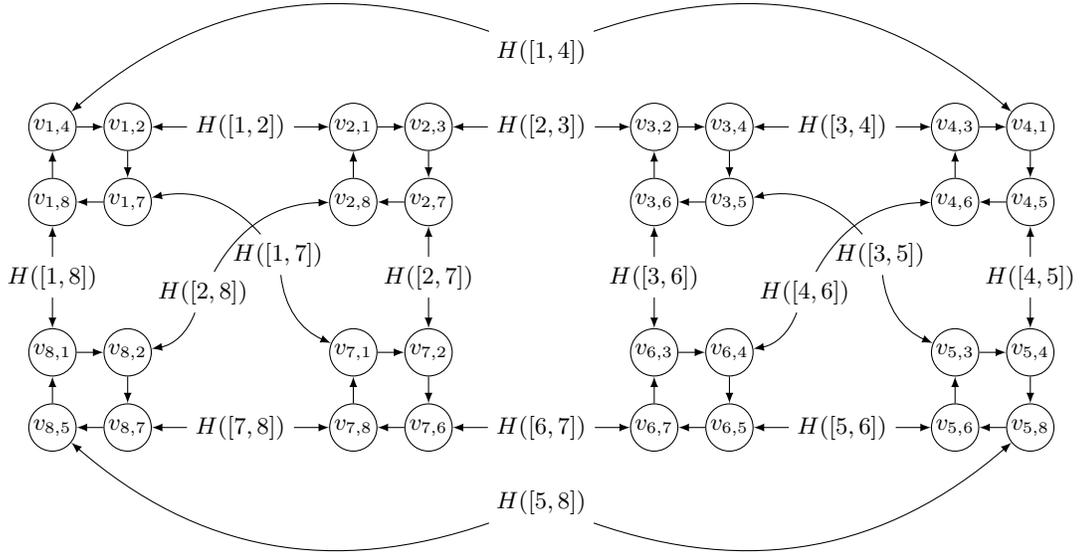
\smallskip

We claim that there is a clique $S\subseteq \{1,\dots,n\}$ of size $k$ if and only if there is $S'\subseteq V'$ satisfying the digraph constraints \eqref{dc} in $G'$ with $w(S')=B=3\Delta nk(k-1)+\Delta k$.\\

Assume there exists a clique $S$ of size $|S|=k$ in $G$. Then the subgraph induced by $S$ contains $\frac{k(k-1)}{2}$ edges.
The set $S'= \left( \cup_{i\in S}C_i \right) \cup \{v_e^\ell\in V_E:e=[i,j]\in E(S),~\ell=1,\dots,6\}$ satisfies the digraph constraints \eqref{dc} in $G'$ and $w(S')=6\Delta n\frac{k(k-1)}{2}+\Delta |S|=3\Delta nk(k-1)+\Delta k=B$.
\smallskip

Conversely, assume there exists $S'\subseteq V'$ satisfying the digraph constraints \eqref{dc} with $w(S')=B=3\Delta nk(k-1)+\Delta k$ for some integer $k\geq 2$. Then $S'$ contains $\frac{k(k-1)}{2}$ gadgets $H(e)$ (because each $H(e)$ is strongly connected) with total weight $6\Delta n$ and $\Delta k$ vertices from $V$ with weight $1$ since the weights of vertices of $V'\setminus V_E$ cannot compensate the weights of one vertex of $V_E$. Due to the digraph constraints, for every gadget $H(e)\subseteq S'\cap V_E$, the 
two circuits $C_i,C_j$ such that $e=[i,j]$ must entirely belong to $S'$.
We construct the set $S=\{i:C_i\subseteq S'\}$. Then $S$ contains exactly $k$ vertices. In addition, the subgraph $G_S$ of $G$ induced by $S$ contains $\frac{k(k-1)}{2}$ edges representing the $\frac{k(k-1)}{2}$ gadgets $H(e)$ in $S'\cap V_E$. We conclude that $G_S$ is a complete graph, so $S$ is a clique of size $k$ in $G$.\hfill\qed
\end{proof}

\begin{corollary}
 {\sc Cardinality SSG} is strongly \textbf{NP}-complete in general digraphs even if the weight of each node is either $a$ or $b$ with $1\leq a<b$  integers.
\end{corollary}

\begin{proof} Using the reduction proposed in the proof of  Theorem \ref{Strong-NPC}, we can see that the weight of each node of $G'$ is either 1 or $\Delta n$ and the size of $S$ is exactly $p=\Delta k +6\frac{k(k-1)}{2}$ where $k$ is the size of the clique. \hfill\qed
\end{proof}

We now prove that Theorem \ref{Strong-NPC} is the best possible complexity result according to degree parameters, that is
either \textsc{SSG} is \textbf{NP}-hard in connected digraphs of maximum degree $2$ or  maximum out-degree $1$ or in-degree $1$,
but admits a pseudo polynomial algorithm for these digraphs or \textsc{SSG} is polynomial for connected digraphs with in-degree and out-degree $2$ for each node.

\begin{lemma}\label{euler}
\textsc{SSG} is polynomial-time solvable  in connected digraphs for which  the  in-degree and the out-degree of each node is $2$.
\end{lemma}

\begin{proof}
Let $I=(G,w,B)$ be an instance of \textsc{SSG} such that $G=(V,A)$ is a connected digraph where the in-degree and the out-degree of each node is exactly 2.
\smallskip

Using Euler-Hierholzer Theorem \cite{BLW76}, we know that $G$ is Eulerian, i.e. there is a circuit visiting each arc of $A$ exactly once. Hence, by the digraph constraints \eqref{dc}, a feasible solution is either the empty set or the whole set of vertices $V$. Since $w$ is non-negative, it follows that $V$ is an optimal solution if and only if $w(V)\leq B$.\hfill\qed

\end{proof}

For the other cases, 
the digraph is a chain or an oriented tree (see Propositions \ref{Prop-NPC} and \ref{Theo:DynProg_Tree}, page \pageref{Prop-NPC}, and Remark \ref{Rem-tree}, page \pageref{Rem-tree}) and then belongs to the class of {\em Directed Acyclic Graphs} (DAG in short).

\section{Directed Acyclic Graphs}\label{gen_prop}


Let us start by some definitions and notions useful in the rest of this section. Given a connected DAG $G=(V,A)$, we say that $v' \in V$ is a {\em descendant} of $v$ iff $v=v'$ or there is a directed path from $v$ to $v'$ in $G$. Let $desc_G(v)$ denote the set of descendants of $v$ in $G$.
An {\em ascendant} of $v$ is a node $u\neq v$ such that $v\in desc_G(u)$. The set of ascendents of $v$ in $G$ is denoted by $asc_G(v)$.
Obviously, $desc_G(v) \cap asc_G(v)=\emptyset$ because $G$ is a DAG. By extension, given $S\subseteq V$, $desc_G(S)=\cup_{v\in S}desc_G(v)$ and $asc_G(S)=\cup_{v\in S}asc_G(v)$. Clearly, a solution $S$ satisfies the digraph constraints \eqref{dc} iff $S=desc_G(S)$.
\smallskip

The {\em kernel} of a subset $S\subseteq V$, denoted by $\kappa(S)$,  is a subset of minimal size such that:
\begin{enumerate}
\item $\kappa(S) \subseteq S$
\item $desc_G(\kappa(S))=desc_G(S)$.
\end{enumerate}

Because $G$ is a DAG, the notion of kernel is well defined and unique (actually, $\kappa(S)$ is the set of sources of the subgraph $G[S]$ induced by $S)$. Note that $\kappa(S)$ is an independent set of $G$. The notion of  $\kappa(S)$ is important because it constitutes in some sense the core of the digraph constraints. The following properties can be easily checked:

\begin{property} \label{property1DAG}
Let $S$ be a subset of a DAG $G=(V,A)$ satisfying  the digraph constraints \eqref{dc}. Let $v\in V$:
\begin{itemize}
\item $v\in S$ iff $desc_G(v)\subseteq S$
\item $v\notin S$ implies $asc_G(v)\cap S=\emptyset$
\end{itemize}
\end{property}

\noindent Now, we show that  \textsc{SSG} and  \textsc{Maximal SSG} in general digraphs can be restricted to DAG.

\begin{lemma}\label{lem}
The resolution (or approximation) of \textsc{SSG} (\textsc{Maximal SSG} resp.) in general digraphs and connected DAG are equivalent.
\end{lemma}
\begin{proof}

Take a digraph $G$, instance of {\sc SSG} ({\sc Maximal SSG} resp.) and replace each strongly connected component by a representative vertex whose weight is the sum of the weights of the nodes that it represents. It is known that  the resulting graph $G'$, called the {\em condensation} of $G$, is a DAG. It is not difficult to see that there is a bijection between the feasible solutions in $G$ and the feasible solutions in $G'$. Moreover, the weight of the solution is preserved.
\hfill\qed
\end{proof}

\begin{corollary}\label{CardinalitySSG-DAG-NPC}
{\sc Cardinality SSG} is strongly \textbf{NP}-complete in DAG even if the weight of each node is either $a$ or $b$ with $1\leq a<b$  integers.
\end{corollary}

\begin{proof} Using both reductions proposed in the proofs of  Theorem \ref{Strong-NPC} and Lemma \ref{lem}, we can see that the weight of each node of $G'$ is either 1 or $6n$ ($\Delta$ and $6\Delta n$ simplified to 1 and $6n$) and the size of $S$ is exactly $p= k +\frac{k(k-1)}{2}$ where $k$ is the size of the clique. \hfill\qed
\end{proof}

\begin{remark}\
\begin{enumerate}
 \item {\sc Cardinality SSG} is polynomial if all the nodes have the same weight $a$.
 \item {\sc Cardinality SSG} is strongly \textbf{NP}-complete if the weight of a node is either $0$ or $a>0$.
 \item {\sc SSG} is polynomial-time solvable when there is a unique weight $a$ or two weights $0$ and $a>0$.
\end{enumerate}

\end{remark}

\begin{proof}\
\begin{enumerate}
 \item In the case of a unique weight $a\geq 0$, if $|V|*a < B$ or $B \neq p*a$, then the answer is no. Otherwise, start from $S=\emptyset$ and iteratively add to $S$ a sink of $G[V \setminus S]$.

\item If the weight of each node is either $0$ or $a>0$, we use the same reduction as in the proof of Corollary \ref{CardinalitySSG-DAG-NPC} and we replace the weights $\Delta$ and $6\Delta n$ by
1 and 0 respectively.

\item In the case of {\sc SSG}, start from $S=\emptyset$ and iteratively add to $S$ a sink of $G[V \setminus S]$.
\hfill\qed
\end{enumerate}
\end{proof}

\begin{lemma}\label{lem_local}
Let $I=(G,w,B)$ be an instance of {\sc Maximal SSG} ({\sc Maximal SSGW}, resp.) such that  $G$ is a DAG.
 $S\subseteq V$ is a feasible solution to {\sc Maximal SSG} if and only if $S$ satisfies \eqref{dc}, \eqref{bd}
and \eqref{busy}.
\begin{equation}
\label{busy} \forall x \in V \setminus S, \quad S \cup \{x\} \mbox{ violates (\ref{dc}) or \eqref{bd}}
\end{equation}
Similarly,  $S\subseteq V$ is a feasible solution to {\sc Maximal SSGW} if and only if $S$ satisfies \eqref{sdc}, \eqref{bd}
and \eqref{sbusy}.
\begin{equation}
\label{sbusy} \forall x \in V \setminus S, \quad S \cup \{x\} \mbox{ violates (\ref{sdc}) or \eqref{bd}}
\end{equation}
\end{lemma}

\begin{proof} The proof is only detailed for {\sc Maximal SSG}. Similar arguments can be used for {\sc Maximal SSGW}.
Let $I=(G,w,B)$ be an instance of \textsc{Maximal SSG} such that $G=(V,A)$ is a DAG.


By definition, a feasible solution to  {\sc Maximal SSG} satisfies \eqref{dc}, \eqref{bd} and \eqref{busy1}.
Since \eqref{busy1} is stronger than \eqref{busy}, one direction of the equivalence (i.e. $\Rightarrow$) holds trivially.


For the other direction, take any $S\subseteq V$ satisfying \eqref{dc}, \eqref{bd} and \eqref{busy}.
By contradiction, suppose there exists $S'\supset S$ such that $S'$ satisfies \eqref{dc} and \eqref{bd}.
Since $G$ is a DAG, the sub-graph $G'$ induced by $S'\setminus S$ is a non-empty DAG. So $G'$ contains a sink $v\in S'\setminus S$.
By definition, $S\cup \{v\}$ satisfies the digraph constraints \eqref{dc}. Moreover, $S\cup \{v\}$ satisfies the budget constraint \eqref{bd} because
$w(S\cup \{v\})\leq w(S')\leq B$. We get a contradiction with \eqref{busy}.  \hfill\qed
\end{proof}

\begin{remark}\label{rem:lem_local}
Lemma \ref{lem_local} is equivalent to check that for a sink $v$ of minimum weight (among the sinks of $G[V\setminus S])$, $w(B)+w(v)>B$.
\end{remark}

\begin{remark}
 Using \eqref{busy}, one can verify if a set $S\subseteq V$ is feasible for {\sc Maximal SSG} in $O(m)$.
\end{remark}


The \textsc{Lazy Bureaucrat Problem} with a common release date and a common deadline has been proved (weakly) \textbf{NP}-hard \cite{GZ08} and  admitting a pseudo-polynomial algorithm \cite{EsfahbodGS03}; recently, a \textbf{FPTAS} was proposed  in \cite{GMP13}.
 Here, we prove that \textsc{Maximal SSG}, a generalization of the lazy bureaucrat problem, is much harder.

\begin{theorem}\label{LazyNP-complete}
\textsc{Maximal SSG} is Strongly \textbf{NP}-hard in connected DAG with a unique sink and 3 distinct weights.
\end{theorem}

\begin{proof}
We propose a Karp reduction from \textsc{Cardinality SSG} proved strongly \textbf{NP}-complete even with 2 distinct weights
in Corollary \ref{CardinalitySSG-DAG-NPC}.

Given an instance $I=(G,w,p,B)$ of \textsc{Cardinality SSG} where $G=(V,A)$ is a DAG, $V=\{v_1,\dots,v_n\}$, $1\leq w(v)\leq B-p$
and  2 distinct weights, we build an instance $I'=(G',w',B',q)$ of the decision version of  \textsc{Maximal SSG} ($I$ is a yes-instance if there exists a feasible solution of weight at most $q$) where $G'=(V',A')$ is a connected DAG with a unique sink as follows:
\begin{itemize}

\item[$\bullet$] $G'=(V',A')$ has $n+p+1$ vertices $V'=V\cup D$ where $D=\{v_{n+i}:i=1,\dots,p+1\}$ and contains $G$ as a subgraph. If $\mathrm{S}$ denotes the set of sinks of $G$, then $A'=A\cup \{(u,v_{n+1}):u\in \mathrm{S}\}\cup \{(v_{i+1},v_i):i=n+1,\dots,n+p\}$.
\smallskip

\item[$\bullet$] $w'(v)=p^2B+pB+w(v)$ for $v\in V$ while $w'(v_{n+i})=p^2B$ for $i=1,\dots,p+1$. Finally,
$B'=p^3B+3p^2B+B-1$.
\smallskip

\item[$\bullet$] $q=p^3B+2p^2B+B$.
\end{itemize}

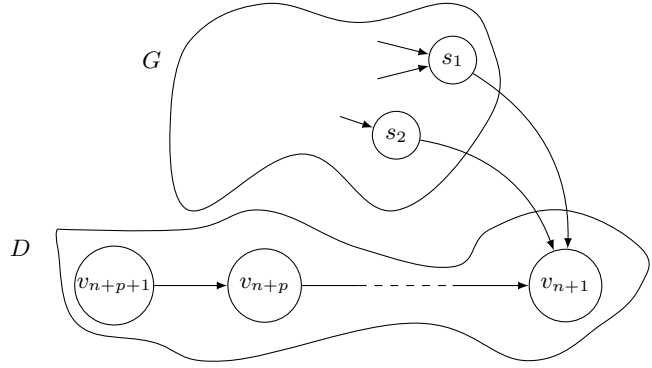
\begin{figure}[h!]
\centering
\begin{tikzpicture}[scale=0.5]
\draw  plot[smooth, tension=.7] coordinates {(-3.5,0.5) (-3,2.5) (-1,3.5) (1.5,3) (4,3.5) (5,2.5) (5,0.5) (2.5,-2) (0,-0.5) (-3,-2) (-3.5,0.5)};

\draw  plot[smooth, tension=.7] coordinates {(-6.5,-2.5) (-3,-2.5) (-1,-2) (1.5,-3) (4,-3.5) (5,-2.5)  (7,-2) (9,-3) (9,-4) (8,-5) (6,-6) (3,-5) (-2,-6) (-4,-5.5) (-6,-5) (-6.5, -3) (-6.5,-2.5)};

\begin{scope}[every node/.style={draw,circle,black}]
\node(s1) at (4,2) {$s_1$};
\node(s2) at (2.5,0) {$s_2$};
\node(2) at (-1,-4) {$v_{n+p}$} ;
\node(k+1) at (7,-4) {$v_{n+1}$} ;
\end{scope}

\begin{scope}[every node/.style={draw,circle,black,inner sep=0.4}]
\node(1) at (-5,-4) {$v_{n+p+1}$};
\end{scope}

\draw[-latex] (1) to (2) ;
\draw (2) to (1.5,-4) ;
\draw[dashed] (1.5,-4) to (4,-4) ;
\draw[-latex] (4,-4) to (k+1) ;
\draw[-latex] (2,2.5) to (s1) ;
\draw[-latex] (2,1.5) to (s1) ;
\draw[-latex] (1,0.5) to (s2) ;
\draw[-latex] (s1) to [bend left] (k+1) ;
\draw[-latex] (s2) to [bend left] (k+1) ;

\node(0) at (-4,2) {$G$} ;

\node(0) at (-7.5,-3) {$D$} ;

\end{tikzpicture}

 \caption{The graph $G'$.}
 \label{fig1}
\end{figure}

Figure \ref{fig1} gives an illustration of this construction. $G'$ is a connected DAG with unique sink and $w'(v)$ can take at most 3 distinct values. These values are positive integers. Obviously, this reduction can be done in polynomial time, and all values $w'(v)$ and $B'$ are upper bounded by a polynomial because $w(v)$ and $B$ are as such. \\

We claim that $I$ is a yes-instance of \textsc{Cardinality SSG}, that is $\exists J\subset V$ with $|J|=p$, $J$ satisfies \eqref{dc} and $w(J)=B$ iff $\exists J'\subset V'$ such that $w'(J')\leq q=p^3B+2p^2B+B$ and $w'(J'')>B'$ for every $J''$ of $V'$ satisfying the digraph constraints and containing $J'$.\\

Clearly, if $I$ is a yes-instance of \textsc{Cardinality SSG}, then by definition  there is $J\subseteq V$ with
$|J|=p$ such that $J$ satisfies \eqref{dc} and $\sum_{v\in J}w(v)=B$. Hence for \textsc{Maximal SSG}, the subset $J'=J\cup\{v_{n+1}\}$ has weight $\sum_{u\in J'}w'(u)=p^2B+p(p^2B+pB)+\sum_{v\in J}w(v)=p^3B+2p^2B+B=q$. This set also satisfies the maximality constraint because using
Property \eqref{busy} of Lemma \ref{lem_local}, we know that the addition to $J'$ of any sink $u$ of the subgraph of $G'$ induced by $V'\setminus J'$  gives $w'(J'\cup\{u\})\geq q+p^2B=p^3B+3p^2B+B>B'$ since $w'(v)\geq p^2B$ for all $v\in V'$. \\

Conversely, let $J'\subset V'$ be a feasible solution of {\sc Maximal SSG} in $G'$ with weight $\sum_{v\in J'}w'(v)\leq q=p^3B+2p^2B+B$. First, let us show that $(i)$ $v_{n+1}\in J'$ and $D\nsubseteq J'$, $(ii)$ this sum $w(J')=q=p^3B+2p^2B+B$  and $(iii)$ $|J'|=p+1$ and $J'\cap D=\{v_{n+1}\}$.

\begin{itemize}
\item For $(i)$. By the maximality constraint and because $v_{n+1}$ is the unique sink of $G'$, we must have  $v_{n+1}\in J'$. Now by contradiction, suppose $D\subseteq J'$. Then, let us prove that $J'\cap V\neq \emptyset$ because otherwise $J'=D$. Since $G=(V,A)$ is a DAG, there exists a sink $u\notin J'$ of $G$. The maximality constraint on $G'$ is not satisfied because  $B'- w'(J')=(p^3B+3p^2B+B-1)-(p+1)p^2B=2p^2B+B-1>w'(u)$. Hence, $u$ should be added to  $J'$. In conclusion $D\cup\{u\}\subseteq J'$ and we deduce $w'(J')\geq w'(D\cup\{u\})\geq (p+1)p^2B+p^2B+pB+1>p^3B+2p^2B+B=q$ which is a contradiction with the initial hypothesis.\\

    \item For $(ii)$. Using $(i)$, we know that there exists $v_{n+\ell}\notin J'$ and $v_{n+\ell-1}\in J'$ for some $\ell\in \{2,\dots,p+1\}$. Hence, the maximality constraint imposes that $v_{n+\ell}$ is a sink in the subgraph of $G'$ induced by $V' \setminus J'$, so  $w'(J')\geq B'-w'(v_{n+\ell})+1=p^3B+2p^2B+B=q$.\\

       \item For $(iii)$. First, observe that $|J'\setminus\{v_{n+1}\}|=p$ because on the one hand if $|J'|\leq p$, then $w'(J')\leq p^2B+ (p-1)(p^2B+pB+B) <p^3B+p^2B-B<q$ and on the other hand, if  $|J'|\geq p+2$, then $w'(J')> p\times p^2B+2(p^2B+pB+1)=p^3B+2p^2B+2pB+2>p^3B+2p^2B+B=q$, because for both cases by item $(i)$ we know $|J'\cap D|\leq p$ and the weights of nodes in $D$ are the smallest of $G'$. These two cases lead to a contradiction with item $(ii)$. Finally, let us prove that $J'\cap D=\{v_{n+1}\}$. Otherwise, $|J'\cap V|\leq p-1$. Since the worst case appears when $|J'\cap V|= p-1$, then $w'(J')< \left((p-1)(p^2B+pB)+(p-1)B\right)+2p^2B=p^3B+2p^2B-B<p^3B+2p^2B+B=q$, which is another contradiction with item $(ii)$.

\end{itemize}

Using $(i)$, $(ii)$ and $(iii)$, and by setting $J=J'\setminus \{v_{n+1}\}$ we deduce $J\subseteq V$, $J$ satisfies the digraph constraints and $|J|=p$ with $w(J)=w'(J')-p(p^2B+pB)-p^2B=q- p^3B-2p^2B=B$. Hence, $I$ is a yes-instance of \textsc{Cardinality SSG}. \hfill\qed
\end{proof}

We now prove that the weak digraph constraints versions of the two problems are as hard to approximate as two variants of the
{\sc Independent Set} problem (\textsc{IS} in short):

\begin{center}
\begin{tabular}{|l|}
\hline
\textsc{IS}\\
\hline
\textsl{Input}: a connected simple graph $G=(V,E)$. \\
\textsl{Output}:  $V'\subset V$ such that no two vertices in $V'$ are joined by an edge in $E$.\\
\textsl{Objective}: maximize $|V'|$.\\
\hline
\end{tabular}
\end{center}

\noindent The {\em independence number} of $G$, denoted by $\alpha(G)$, is the maximum size of an independent set in
$G$. \textsc{IS} is known to be \textbf{APX}-complete in cubic graphs \cite{AlimontiK00}, \textbf{NP}-hard in planar graphs with maximum degree $\Delta(G)\leq 3$ \cite{GJ79} and no polynomial algorithm can approximately solve it within the ratio
$n^{\varepsilon-1}$ for any $\varepsilon\in (0;1)$, unless \textbf{P}$=$\textbf{ZPP}
\cite{hastad1996clique}.

\medskip

The second problem is the {\sc Minimum Independent Dominating Set} problem (\textsc{ISDS} in short) also known as the {\sc Minimum Maximal Independent Set} problem (Problem [GT2], see comment page 190 in \cite{GJ79}).

\begin{center}
\begin{tabular}{|l|}
\hline
\textsc{ISDS}\\
\hline
\textsl{Input}: a connected simple graph $G=(V,E)$. \\
\textsl{Output}:  $V'\subset V$ such that no two vertices in $V'$ are joined by an edge in $E$\\
and for all $u\in V\setminus V'$, there is a $v\in V'$ for which $(u,v)\in E$.\\
\textsl{Objective}: minimize $|V'|$.\\
\hline
\end{tabular}
\end{center}

\medskip

\noindent The {\em independent domination number} of $G$, denoted by $i(G)$  is the minimum size
of an independent dominating set in $G$. Obviously, $i(G)\leq \alpha(G)$.
\textsc{ISDS} is known to be \textbf{NP}-hard \cite{GJ79}, even for planar cubic graphs \cite{Manlove99} and it is \textbf{APX}-complete
for graphs of maximum degree $3$ \cite{ChlebikC08}. Furthermore, it is also very hard from an approximation point of view, since no polynomial algorithm
can approximately solve it within the ratio $n^{1-\varepsilon}$ for any $\varepsilon\in (0;1)$, unless \textbf{P}$=$\textbf{NP} \cite{Halldorsson93a}.

\begin{theorem}\label{Theo:Reduc-SSGW}
There is a polynomial reduction preserving approximation from
\begin{enumerate}
\item \textsc{IS} in general graphs to \textsc{SSGW} in DAG with maximum in-degree 2.
\item \textsc{ISDS} in general graphs to \textsc{Maximal SSGW} in DAG with maximum in-degree 2.
\end{enumerate}
\end{theorem}
\begin{proof}
Let $G$ be a connected graph where $G=(V,E)$  with
$V=\{v_1,\dots,v_n\}$. We construct a corresponding instance $I=(G'=(V',A'),w,B)$ of \textsc{SSGW} and \textsc{Maximal SSGW} as follows:
Let $G'=(V',A')$ be a digraph defined by $V'=V\cup V_E$ where $V_E=\{v_e:e\in E\}$ and $A'=\{(v_i,v_{e}): e=[v_i,v_j]\in E\}$.
Set $w(v_i)=1$ for $i=1,\dots,n$ and $w(v_e)=n+1$ for all $v_e\in V_E$ and $B=n$.
Clearly, $G'$ is a connected DAG with maximum in-degree 2.
An example of such reduction is given in Figure \ref{fig2Theo:Reduc-SSGW}.
\smallskip

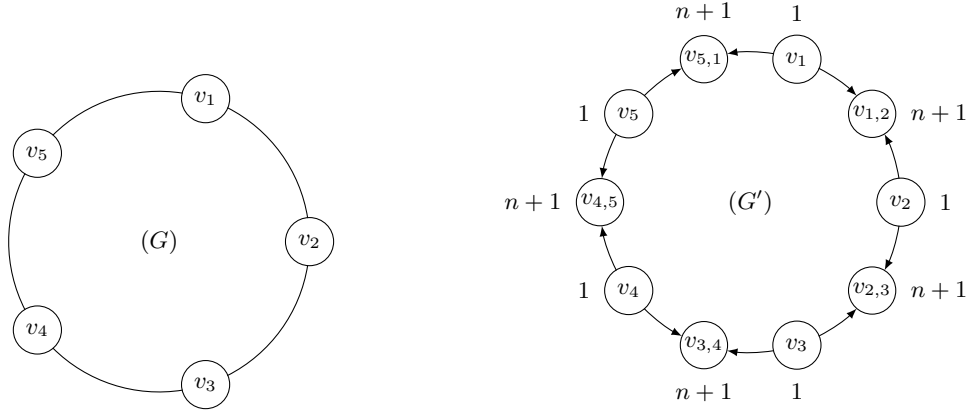
\begin{figure}[h!]
\centering
\begin{tikzpicture}
\node(0) at (0,0) {$(G)$} ;

\def \n {5}
\def \radius {2cm}
\def \margin {9} 

\node[draw, circle] at ({360/\n * (1 )}:\radius) {$v_1$};
\node[draw, circle] at ({360/\n * (2 )}:\radius) {$v_5$};
\node[draw, circle] at ({360/\n * (3 )}:\radius) {$v_4$};
\node[draw, circle] at ({360/\n * (4 )}:\radius) {$v_3$};
\node[draw, circle] at ({360/\n * (5 )}:\radius) {$v_2$};

\foreach \s in {1,...,\n}
{
  \draw[-] ({360/\n * (\s )+\margin}:\radius)
    arc ({360/\n * (\s )+\margin}:{360/\n * (\s+1)-\margin}:\radius);
}
\end{tikzpicture}
\hspace{2cm}
\begin{tikzpicture}
\node(0) at (0,0) {$(G')$} ;

\def \n {10}
\def \radius {2cm}
\def \margin {9} 

\node[draw, circle,inner sep=0.8](1) at ({360/10 * (1 )}:\radius) {$v_{1,2}$};
\node[draw, circle](2) at ({360/10 * (2)}:\radius) {$v_{1}$};
\node[draw, circle,inner sep=0.8](3) at ({360/10 * (3 )}:\radius) {$v_{5,1}$};
\node[draw, circle](4) at ({360/10 * (4 )}:\radius) {$v_5$};
\node[draw, circle,inner sep=0.8](5) at ({360/10 * (5 )}:\radius) {$v_{4,5}$};
\node[draw, circle](6) at ({360/10 * (6 )}:\radius) {$v_{4}$};
\node[draw, circle,inner sep=0.8](7) at ({360/10 * (7 )}:\radius) {$v_{3,4}$};
\node[draw, circle](8) at ({360/10 * (8 )}:\radius) {$v_3$};
\node[draw, circle,inner sep=0.8](9) at ({360/10 * (9 )}:\radius) {$v_{2,3}$};
\node[draw, circle](10) at ({360/10 * (10 )}:\radius) {$v_2$};

\node at (1) [right=0.4] {$n+1$};
\node at (2) [above=0.4] {$1$};
\node at (3) [above=0.4] {$n+1$};
\node at (4) [left=0.4] {$1$};
\node at (5) [left=0.4] {$n+1$};
\node at (6) [left=0.4] {$1$};
\node at (7) [below=0.4] {$n+1$};
\node at (8) [below=0.4] {$1$};
\node at (9) [right=0.4] {$n+1$};
\node at (10) [right=0.4] {$1$};

\foreach \s in {1,...,\n}
{
    \ifthenelse{\isodd{\s}}
    {\draw[<-, >=latex] ({360/\n * (\s )+\margin}:\radius) arc ({360/\n * (\s )+\margin}:{360/\n * (\s+1)-\margin}:\radius);}
    {\draw[->, >=latex] ({360/\n * (\s )+\margin}:\radius) arc ({360/\n * (\s )+\margin}:{360/\n * (\s+1)-\margin}:\radius);};

}
\end{tikzpicture}
\caption{Example of reduction.}
 \label{fig2Theo:Reduc-SSGW}
\end{figure}

We claim that $S$ is a maximal independent set of $G$ if and only if the same set of vertices in digraph $G'$ is a feasible solution to \textsc{SSGW} (\textsc{Maximal SSGW} resp.).\\

$S$ is a maximal independent set of $G$ iff this subset satisfies the weak digraph constraints in $G'$ (otherwise, $v_e$ for some $e=[x,y]$ with $x,y\in S$ should be added to $S)$ and using Lemma \ref{lem_local}, we can not add a new vertex  because  on the one hand, $w(S)=|S|\leq n$ and on the other hand, either $S+v$ does not satisfies the weak  digraph constraints for $v\notin S$ by maximality of $S$ or $w(S+v_e)>n$ for some $e\in E$.

\bigskip

For item $1$, the result follows from the definition of \textsc{IS} while for item $2$ it is the the definition of \textsc{ISDS}.\hfill\qed
\end{proof}

\begin{corollary}\label{Cor:Theo-Reduc-SSGW}
In DAG with maximum in-degree is 2, we have:
\begin{enumerate}
\item \textsc{SSGW} and \textsc{Maximal SSGW} are \textbf{APX}-hard even if the maximum out-degree is 3.

\item For every $\varepsilon \in (0,1/2)$, \textsc{SSGW} (\textsc{Maximal SSGW} resp.) is  not $n^{\varepsilon-\frac{1}{2}}$, unless \textbf{P}$=$\textbf{ZPP} $(n^{\frac{1}{2}-\varepsilon}$ unless \textbf{P}$=$\textbf{NP} resp.).
\end{enumerate}
\end{corollary}
\begin{proof}
For item $1$, it is a consequence of Theorem \ref{Theo:Reduc-SSGW} together with the results of \cite{AlimontiK00,ChlebikC08}.\\

For item $2$, we use the negative results given in \cite{hastad1996clique,Halldorsson93a} and the reduction of Theorem \ref{Theo:Reduc-SSGW}
with $|V(G')|= |E(G)|+|V(G)|\leq |V(G)|^2$.\hfill\qed
\end{proof}

\begin{lemma}\label{lem_Tournament}
\textsc{SSG} and \textsc{Maximal SSG} are polynomial-time solvable in tournaments.
\end{lemma}
\begin{proof}
Let $G$ be a tournament digraph. We may assume that $G$ is acyclic by Lemma \ref{lem}.
Thus $G$ contains a unique Hamiltonian  path \cite{BP93}.
We denote it by $\mathcal{H}=\{(v_i,v_{i+1}):i=1,\dots,n-1\}$
where $n$ is the number of vertices of $G$. Due to the digraph constraints,
a feasible solution is, either the empty set, or a subset $\{(v_i,v_{i+1}): i=k,\dots,n-1\}$
for some $k\in \{1,\dots,n-1\}$ because $G$ is an acyclic tournament and then $(v_i,v_j)\notin A$ for $j<i$. Once $\mathcal{H}$ is found (this can be done in polynomial time), we can make a binary
search on $\mathcal{H}$ to find a solution if it exists in $O(n\log n)$. \hfill\qed
%
%
%
%
%
%
%
%

\end{proof}

\section{Oriented Trees}\label{tree}

\begin{proposition}\label{Prop-NPC}
The four following problems are \textbf{NP}-hard in out-rooted trees and in in-rooted trees:
\begin{enumerate}
 \item \textsc{SSG} and \textsc{SSGW}.
 \item \textsc{Maximal SSG} and \textsc{Maximal SSGW}.
\end{enumerate}
\end{proposition}

\begin{proof}
We only prove the case of out-rooted trees (for in-rooted trees, we reverse the orientation of each arc).\\

For item 1, we show the \textbf{NP}-hardness using a reduction from {\sc Subset Sum} (Problem [SP13], page 223 in \cite{GJ79}) known to be (weakly) \textbf{NP}-complete. This problem is described as follows:
\smallskip

\begin{center}
\begin{tabular}{|l|}
\hline
\textsc{Subset Sum (SS)}\\
\hline
\textsl{Input}: a finite set $X$, a size $s(x)\in \mathbb{Z}^+$ for each $x \in X$ and a positive integer $B$. \\
\textsl{Question}:  is there a subset $S\subseteq X$ such that $s(S)=\sum_{x\in S}s(x)=B$?\\
\hline
\end{tabular}
\end{center}


Let $I=(X,s,B)$ be an instance of \textsc{SS}. We polynomially construct a corresponding instance $I'=(T,w,B,k)$ of the  decision version  \textsc{SSG} and \textsc{SSGW} where $k$ is an  integer. Let $T=(V,A)$ be a digraph defined by $V=X\cup \{r\}$ and $A=\{(v,r): v\in X\}$. Clearly, $T$ is an out-rooted tree. The weight function is $w(v)=s(v)$ for all $v\in X$ and $w(r)=0$. Finally, we set $k=B$.
\smallskip

It is easy to show for  \textsc{SSG} (\textsc{SSGW} resp.) that $S$ is a solution of {\sc Subset Sum} if and only if $S\cup \{r\}$ satisfies the digraph constraints \eqref{dc} (the weak digraph constraints \eqref{sdc} resp.)  and $w(S\cup \{r\})\geq k$.
\bigskip

For item 2, we prove the \textbf{NP}-hardness using a reduction from the {\sc Lazy Bureaucrat Problem} with common deadlines and release dates. The decision version of this problem has been shown \textbf{NP}-complete in \cite{GZ08} and it can be described by:
\smallskip

\begin{center}
\begin{tabular}{|l|}
\hline
{\sc Decision Lazy Bureaucrat}\\
\hline
\textsl{Input}: a finite set $X$, a size $s(x)\in \mathbb{Z}^+$ for each $x \in X$, positive integers $B$ and $k\leq B$. \\
\textsl{Question}: is there a subset $S\subseteq X$ such that $s(S)=\sum\limits_{x\in S}s(x)\leq k$ and $\forall x\notin S$, $s(S)+s(x)>B$?\\
\hline
\end{tabular}
\end{center}


Let $I=(X,s,B,k)$ be an instance of {\sc Decision Lazy Bureaucrat}. We construct an instance $I'=(T,w,B,k)$ in the same way as for item 1.
Clearly,  there is a subset $S\subseteq X$ with $s(S)=\sum_{x\in S}s(x)\leq k$ and $\forall x\notin S$, $s(S)+s(x)>B$ if and only if $S\cup \{r\}$ satisfies the (weak resp.) digraph constraints with $w(S\cup \{r\})\leq k$ and  $\nexists S'\supset S$ that satisfies the (weak resp.) digraph constraints with $w(S'\cup \{r\})\leq B$.\hfill\qed
\end{proof}

\begin{remark}\label{Rem-tree}
The reduction of Proposition \ref{Prop-NPC} can also be modified in order to get $w(v)>0$ for every vertex $v$. Moreover, we can slightly modify the construction in order to obtain a binary tree or a chain.
\end{remark}

We now present some  dynamic programs for solving the different problems defined in Section \ref{pbs}
in the class of trees.
\smallskip

Beforehand, we introduce some notations on trees. Let $T=(V,A)$ be a directed tree. Let us fix any vertex $r(T)\in V$ as the root of the underlying tree $T$.
\smallskip

For a node $v \in V$, we denote by $fa(v)$ its father and by $ch(v)$ its set of children. The root has no father and its neighbors are its children. For a node $v$ that is not $r(T)$, its father $fa(v)$ is the first node of the unique path from $v$ to $r(T)$ (with $v \neq fa(v)$). Note that this path is not necessarily directed, i.e. an arc can be traversed from head to tail or the other way. The children set of $v$, denoted by $ch(v)$, is defined as $N_T(v) \setminus \{fa(v)\}$.
\smallskip


We partition the set $ch(v)$ of children of $v$ into two sets  $ch^+(v)$ and $ch^-(v)$. The set
$ch^+(v)$ contains the children of $v$ in $T$ that leave $v$, i.e. $ch^+(v) = \{u \in ch(v) : (v,u) \in A\}$. The set $ch^-(v)$ contains the children of $v$ in $T$ that arrive in $v$, i.e. $ch^-(v) = \{u \in ch(v) : (u,v) \in A\}$.
\smallskip

For a child $i$ of the root $r(T)$, let $V_i$ denote the nodes accessible from $i$ without passing through $r(T)$. That is, $V_i$ consists of vertex $i$, the children of $i$, the children of these children, etc. Let $T_i$ denote the tree induced by $V_i$ in which $i$ is the root.


\begin{figure}[h!]
\centering
 \begin{tikzpicture}
    \tikzstyle{level 1}=[sibling distance=40mm]
    \tikzstyle{level 2}=[sibling distance=15mm]
    \tikzstyle{every node}=[circle,draw]

    \node (A) {$v_1$}
        child { node (B) {$v_2$}
                child { node (C) {$v_4$} edge from parent[->]}
                child { node (D) {$v_5$} edge from parent[<-]}
              edge from parent[->]}
        child {
            node (E) {$v_3$}
            child { node (F) {$v_6$} edge from parent[->]}
            child { node (G) {$v_7$} edge from parent[<-]}
            child { node (H) {$v_8$} edge from parent[<-]}
        edge from parent[<-]};

 \end{tikzpicture}
\caption{Example of a directed tree.}
\label{ex_tree}
\end{figure}
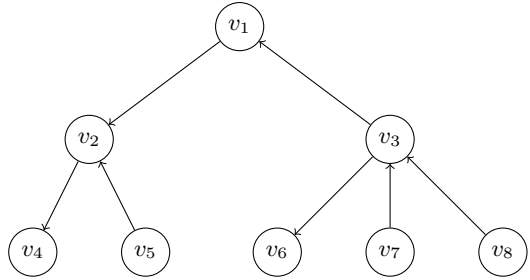

An example of a directed tree $T=(V,A)$ is given in figure \ref{ex_tree}. Let us fix arbitrarily the vertex $v_1$ as the root $r(T)$ of the tree $T$. Then $fa(v_1)=\emptyset$, $ch(v_1)=\{v_2,v_3\}$, $ch^+(v_1)=\{v_2\}$ and $ch^-(v_1)=\{v_3\}$. Regarding vertex $v_3$, $fa(v_3)=\{v_1\}$, $ch(v_3)=\{v_6,v_7,v_8\}$, $ch^+(v_3)=\{v_6\}$ and $ch^-(v_3)=\{v_7,v_8\}$. Also, $ch(v_i)=\emptyset$ for $i\in \{4,5,6,7,8\}$.
Moreover, $V_{v_2}=\{v_2,v_4,v_5\}$ and $T_{v_2}$ is the sub-tree induced by $V_{v_2}$. In the same way, $T_{v_3}$ is the sub-tree induced by $V_{v_3}=\{v_3,v_6,v_7,v_8\}$. Observe that $V=\{r(T)\}\cup_{i\in ch(r(T))}V_i$.

\begin{proposition}\label{Theo:DynProg_Tree}
\textsc{SSG} can be solved using dynamic programming in oriented trees.
\end{proposition}

\begin{proof}
Let $I=(T,w,B)$ be an instance of \textsc{SSG} where $T=(V,A)$ is a directed tree with a root $r(T)\in V$. Given an integer $b \in \{0,...,B\}$, let $R^+(T,b)$ ($R^-(T,b)$ resp.) be the boolean defined by $R^+(T,b)=True$ ($R^-(T,b)=True$ resp.) if and only if there exists $S\subseteq V$  satisfying \eqref{dc} in the tree $T$ with $r(T)\in S$ ($r(T)\notin S$ resp.) and $w(S)=b$. Let $R(T,b)=R^+(T,b)\vee R^-(T,b)$. Then $R(T,b)$ is $True$ if and only if there exists $S\subseteq V$ satisfying \eqref{dc} and $w(S)=b\leq B$, so $S$ also satisfies \eqref{bd}. Hence it is feasible for SSG with the weight $w(S)=b$.
We define $R^+(T,b)$ and $R^-(T,b)$ recursively as follows\footnote{The formula "$==$" is the boolean test of equality.}:
\smallskip



$$
R^+(T,b) = \left\{
    \begin{array}{ll}
        b==w(r(T)), & \mbox{when } V=\{r(T)\}, \\
        \bigwedge\limits_{k\in ch^+(r(T))} R^+(T_k,a_k) \bigwedge\limits_{\ell\in ch^-(r(T))}R(T_\ell,b_\ell), &  \mbox{when } |V|>1 \mbox{ and where }\\
        & a_k\geq 0, \forall k\in ch^+(r(T)),\\
        & b_\ell\geq 0, \forall \ell\in ch^-(r(T)) \mbox{ and }\\
        & \sum\limits_{k\in ch^+(r(T))}a_k+\sum\limits_{\ell\in ch^-(r(T))}b_\ell=b-w(r(T))
    \end{array}
\right.
$$

$$
R^-(T,b) = \left\{
    \begin{array}{ll}
        b==0, & \mbox{when } |V|=1, \\
        \bigwedge\limits_{k\in ch^-(r(T))} R^-(T_k,c_k) \bigwedge\limits_{\ell\in ch^+(r(T))}R(T_\ell,d_\ell), &  \mbox{when } |V|>1 \mbox{ and where }\\
        & c_k\geq 0, \forall k\in ch^-(r(T)),\\
        & d_\ell\geq 0, \forall \ell\in ch^+(r(T)) \mbox{ and }\\
        & \sum\limits_{k\in ch^+(r(T))}c_k+\sum\limits_{\ell\in ch^-(r(T))}d_\ell=b
    \end{array}
\right.
$$

%
%
%

and $R(T,b)=R^+(T,b)\vee R^-(T,b)$.
\smallskip

Let us prove that $R$ is well-defined or equivalently that $R^+$ and $R^-$ are well-defined by induction on $|V|$.
If $|V|=1$ then  $V=\{r(T)\}$. In this case, the unique feasible solution for $R^-$ is the empty set, so $R^-(T,0)=True$ and otherwise $R^-(T,b)=false$ for $b\neq 0$ and $b\leq B$. Regarding $R^+$,   the unique solution is reduced to the vertex $\{r(T)\}$, then $R^+(T,w(r(T)))=True$ and $R^+(T,b)=false$ for $b\neq w(r(T))$ and $b\leq B$.
\smallskip

Now, assume that $|V|\geq 2$. Then $r(T)$ has at least one child.
Suppose that $R(T,b)=True$, i.e. $R^+(T,b)=True$ or $R^-(T,b)=True$
for some $b>0$ (the case $b=0$ corresponds to the empty solution).
\smallskip

\begin{itemize}
 \item If $R^+(T,b)=True$ then there exists $S\subseteq V$ satisfying \eqref{dc} in $T$ with $r(T)\in S$ and $w(S)=b\leq B$. It follows that $k=r(T_k)\in S$ for all $k\in ch^+(r(T))$ because of \eqref{dc}. In addition, $S\cap V_k$ necessarily satisfies   \eqref{dc} in the sub-tree $T_k$. By setting $a_k=w(S\cap V_k)$, we conclude that $R^+(T_k,a_k)=True$ for all $k\in ch^+(r(T))$. Moreover, for all $\ell\in ch^-(r(T))$,
$S\cap V_\ell$ satisfies \eqref{dc} in $T_\ell$. By setting $b_\ell=w(S\cap V_\ell)$, we get that $R(T_\ell,b_\ell)=True$ for all $\ell\in ch^-(r(T))$. Finally,  $S=\{r(T)\}\bigcup_{k\in ch^+(r(T))}(S\cap V_k)\bigcup_{\ell\in ch^-(r(T))}(S\cap V_\ell)$, so $w(S)=w(r(T))+\sum\limits_{k\in ch^+(r(T))}a_k+\sum\limits_{\ell\in ch^-(r(T))}b_\ell=b$ which is consistent with the definition of $R^+$.
\smallskip

 \item If $R^-(T,b)=True$ then there is $S\subseteq V$ satisfying \eqref{dc} in $T$ with $r(T)\notin S$ and $w(S)=b\leq B$. Hence, for all $k\in ch^-(r(T))$, $k\notin S\cap V_k$ because of \eqref{dc}. Using the fact that $S\cap V_k$ satisfies \eqref{dc} in the sub-tree $T_k$, it follows that $R^-(T_k,c_k)=True$ where $c_k=w(S\cap V_k)$ for all $k\in ch^-(r(T))$.
Moreover, for all $\ell\in ch^+(r(T))$,
 $S\cap V_\ell$ satisfies \eqref{dc} in the sub-tree $T_\ell$, so $R(T_\ell,w(S\cap V_\ell))=True$.
 The result follows by setting $d_\ell=w(S\cap V_\ell)$ for all $\ell \in ch^+(r(T))$.
\end{itemize}
\medskip

Conversely, we denote by $S_i\subseteq V_i$ a solution satisfying $R(T_i,w(S_i))=True$  for $i \in ch(r(T))$. 
\smallskip

\begin{itemize}
 \item If for all $k\in ch^+(r(T))$, $R^+(T_k,a_k)=True$ with $a_k\geq 0$,
 we know that for every $\ell \in ch^-(r(T))$, there exists $b_\ell \geq 0$ such that $R(T_\ell,b_\ell)=True$ (in the worst case, choose $b_\ell=0$ because $R^-(T_\ell,0)$ is always $True$); then we have
 $\{r(T)\}\bigcup_{k\in ch^+(r(T))}S_k \bigcup_{i\in I}S_i$ satisfies \eqref{dc} in $T$. By setting $b=w(r(T))+\sum_{k\in ch^+(r(T))}a_k+\sum_{\ell\in ch^-(r(T))}a_\ell$, we get that $R^+(T,b)=True$.  Otherwise (if there exists $k\in ch^+(r(T))$ such that $R^+(T_k,a_k)=False$, $\forall a_k\geq 0$), we get that  $R^+(T,b)=False$ for every $b\geq 0$ because of \eqref{dc}. Indeed, the root $r(T)$ can not belong to a feasible solution in $T$ when it has a child $k\in ch^+(r(T))$ which is not in this solution.
\smallskip

\item Moreover, if for all $k\in ch^-(r(T))$, $R^-(T_k,c_k)=True$ with $c_k\geq 0$, we know that for every $\ell \in ch^+(r(T))$, there exists $d_\ell \geq 0$ such that $R(T_\ell,d_\ell)=True$ (in the worst case, choose $d_\ell=0$ because $R^-(T_\ell,0)$ is always $True$); then we have
$\{r(T)\}\bigcup_{k\in ch^-(r(T))}S_k \bigcup_{\ell\in ch^+(r(T))}S_\ell$ satisfies \eqref{dc} in $T$. Let $b=\sum_{k\in ch^-(r(T))}c_k+\sum_{\ell\in ch^+(r(T))}d_\ell$. Then $R^-(T,b)=True$.  Otherwise, if there exists $k\in ch^-(r(T))$ such that $R^-(T_k,a_k)=False$, $\forall a_k\geq 0$, then there is no solution $S_k$ satisfying \eqref{dc} in $T_k$ with $w(S_k)=a_k$ and such that $k\notin S_k$. Hence, for every solution $S$ of $T$, $S$ must contain $k$, and by \eqref{dc}, it must contain $r(T)$, so $R^-(T,b)=False$ for every $b>0$.
\end{itemize}
\smallskip

The value of $R$ can easily be deduced from $R^+$ and $R^-$. The induction is proved.
\bigskip

$R$ can be computed in $O(nB^2)$. Indeed, we can construct two tables, one for $R^+$ and the other one for $R^-$. In both tables, the columns are valued by the integers $0,1,...,B$ and the lines contain sub-trees constructed as follows. We start by adding, in both tables, one line per leaf of $T$ and for each graph $T_\ell$ induced by leaf $\ell$, set $R^-(T_\ell,b)=True$ if and only if $b=0$, and $R^+(T_\ell,b)=True$ if and only if $b=w(\ell)$.
Then, add new lines containing new sub-trees with the following algorithm.

\begin{enumerate}
 \item If there exist two sub-trees $T_i$ and $T_{i'}$ induced by $V_i$ and $V_i'$  respectively with $fa(i)=fa(i')$ whose lines have already been created in the tables then
 \begin{enumerate}
  \item \label{1a}
  If there exists another child $i''\neq i,i'$ such that $fa(i'')=fa(i)=fa(i')$ then choose $i$ and $i'$ such that both are either in $ch^+(fa(i))$ or in $ch^-(fa(i))$ and add a new line associated with the forest $T_{i,i'}$ induced by $V_i\cup V_i'$.
   \begin{enumerate}

     \item If $i,i'\in ch^+(fa(i))$ then
    \begin{itemize}
     \item $R^+(T_{i,i'},b)=True$ for $b\in \{0,...,B\}$ if and only if there exist $b_i,b_{i'}\in \{0,...,b\}$ such that $b=b_i+b_{i'}$ and $R^+(T_{i},b_i)=R^+(T_{i'},b_{i'})=True$.
     \item $R^-(T_{i,i'},b)=True$ for $b\in \{0,...,B\}$ if and only if there exist $b_i,b_{i'}\in \{0,...,b\}$ such that $b=b_i+b_{i'}$ and $R(T_{i},b_i)=R(T_{i'},b_{i'})=True$.
    \end{itemize}

   \item If $i,i'\in ch^-(fa(i))$ then
    \begin{itemize}
     \item $R^+(T_{i,i'},b)=True$ for $b\in \{0,...,B\}$ if and only if there exist $b_i,b_{i'}\in \{0,...,b\}$ such that $b=b_i+b_{i'}$ and $R(T_{i},b_i)=R(T_{i'},b_{i'})=True$.
     \item $R^-(T_{i,i'},b)=True$ for $b\in \{0,...,B\}$ if and only if there exist $b_i,b_{i'}\in \{0,...,b\}$ such that $b=b_i+b_{i'}$ and $R^-(T_{i},b_i)=R^-(T_{i'},b_{i'})=True$.
    \end{itemize}

   \end{enumerate}

  \item Else (there is no other child $i''\neq i,i'$ with the same father $fa(i'')=fa(i)=fa(i')$), add a new line associated with the tree $T_{fa(i)}$ induced by $V_i\cup V_i'\cup \{fa(i)\}$.
    \begin{enumerate}
     \item If $i,i'\in ch^+(fa(i))$ then
     \begin{itemize}
     \item $R^+(T_{fa(i)},b)=True$ for $b\in \{0,...,B\}$ if and only if there exist $b_i,b_{i'}\in \{0,...,b\}$ such that $b=b_i+b_{i'}+w(fa(i))$ and $R^+(T_{i},b_i)=R^+(T_{i'},b_{i'})=True$.
     \item $R^-(T_{fa(i)},b)=True$ for $b\in \{0,...,B\}$ if and only if there exist $b_i,b_{i'}\in \{0,...,b\}$ such that $b=b_i+b_{i'}$ and $R(T_{i},b_i)=R(T_{i'},b_{i'})=True$.
     \end{itemize}

     \item If $i,i'\in ch^-(fa(i))$ then
     \begin{itemize}
      \item $R^+(T_{fa(i)},b)=True$ for $b\in \{0,...,B\}$ if and only if there exist $b_i,b_{i'}\in \{0,...,b\}$ such that $b=b_i+b_{i'}+w(fa(i))$ and $R(T_{i},b_i)=R(T_{i'},b_{i'})=True$.
     \item $R^-(T_{fa(i)},b)=True$ for $b\in \{0,...,B\}$ if and only if there exist $b_i,b_{i'}\in \{0,...,b\}$ such that $b=b_i+b_{i'}$ and $R^-(T_{i},b_i)=R^-(T_{i'},b_{i'})=True$.
     \end{itemize}

     \item \label{da} If $i\in ch^+(fa(i))$ and $i'\in ch^-(fa(i))$ then
      \begin{itemize}
      \item $R^+(T_{fa(i)},b)=True$ for $b\in \{0,...,B\}$ if and only if there exist $b_i,b_{i'}\in \{0,...,b\}$ such that $b=b_i+b_{i'}+w(fa(i))$ and $R^+(T_{i},b_i)=R(T_{i'},b_{i'})=True$.
     \item $R^-(T_{fa(i)},b)=True$ for $b\in \{0,...,B\}$ if and only if there exist $b_i,b_{i'}\in \{0,...,b\}$ such that $b=b_i+b_{i'}$ and $R(T_{i},b_i)=R^-(T_{i'},b_{i'})=True$.
      \end{itemize}

     \item If $i'\in ch^+(fa(i))$ and $i\in ch^-(fa(i))$ then permute $i$ and $i'$ and go to \ref{da}.

    \end{enumerate}

 \end{enumerate}
 \item Else, there exists a sub-graph  $T_i$ induced by $V_i$ whose line has already been created in the tables and such that there is no other child $i'\neq i$ with the same father $fa(i)$, and the sub-tree whose root is the father  $fa(i)$ (i.e. the sub-tree induced by $V_{fa(i)}$) is not  created yet in the tables. Then add a new line containing the graph $T_{fa(i)}$ induced by $V_i\cup \{fa(i)\}$ and set:
 \begin{itemize}
      \item $R^-(T_{fa(i)},b)=True$ if and only if $R^-(T_i,b)=True$ when $i\in ch^-(fa(i))$,  $R(T_i,b)=True$ otherwise.
      \item $R^+(T_{fa(i)},b)=True$ if and only if $R^+(T_i,b-w(fa(i)))=True$ when $i\in ch^-(fa(i))$,  $R(T_i,b-w(fa(i)))=True$ otherwise.
     \end{itemize}

\end{enumerate}
\smallskip


In case \ref{1a}, $T_{i,i'}$ is a forest and not a tree, but we use the boolean $R(T_{i,i'},b)$ which is defined on a tree $T$ and a nonnegative integer $b$. However this has no incidence on the construction of the tables $R^+$ and $R^-$. Indeed, we can transform the forest $T_{i,i'}$ into a tree $T_{ii'}$ by adding a fictitious vertex $fv$. If in $T$ the arcs go from $fa(i)$ to $i$ and $i'$ then add in $T_{ii'}$ the arcs $(fa(i),fv)$, $(fv,i)$, $(fv,i')$, else add these arcs in the opposite direction.
\smallskip

%
%
%
%
%

 The filling of each line can be realized in $O(B^2)$ because of the comparison of each value of line $T_i$ with each value of line $T_{i'}$ in the worst case. We continue to create new lines in both tables in parallel until getting $T$, so we create at most $2n$ lines.
 Then the answer to SSG is yes if and only if $R(T,B)=True$ (i.e. if $R^-(T,B)=True$ or $R^+(T,B)=True$). In case of answer yes, the solution can be obtained by creating another table containing a feasible solution per cell valued True. In the lines containing a leaf, the solution is the empty set if the corresponding $R^-$ is $True$ or the leaf if the corresponding $R^+$ is $True$. Otherwise, the solution of each current $True$ cell in the other lines is the union of the solutions of the sub-trees that compose $T$ and we add $\{r(T)\}$ if and only if the corresponding $R^+$ is $True$ and $T$ is not the union of a subset of children of $r(T)$.

\hfill\qed

\end{proof}

\begin{remark}
 If $T$ is an out-rooted (in-rooted resp.) tree, then it is easy to define the root $r(T)$ as the unique sink (source resp.) of $T$. The cases of out-rooted and in-rooted trees were previously treated in \cite{johnson1983knapsacks} where a dynamic programming algorithm was proposed. Proposition \ref{Theo:DynProg_Tree} is a generalization to any directed tree.
\end{remark}

\begin{remark}\label{rem:forest}
 Proposition \ref{Theo:DynProg_Tree} also holds in the class of forests.
\end{remark}

\begin{proposition}\label{dynprog_lazy}
\textsc{Maximal SSG} can be solved using dynamic programming in oriented trees.
\end{proposition}

\begin{proof}
 Let $I=(T,w,B)$ be an instance of \textsc{Maximal SSG} where $T=(V,A)$ is a tree. We rename the vertices $v_1,\dots,v_n$ of $V$ in such a way that for all $i\in \{1,...,n\}$, $v_i$ is the sink of minimum weight among the sinks of the subgraph induced by the set of vertices $V\setminus \{v_1,...,v_{i-1}\}$. Hence, the root is necessarily $v_n$.
  Wlog., assume $w(V)>B$, since otherwise $V$ is an optimal solution.
 \smallskip

 Let $k\geq 1$.  We denote by $T_k$ the forest induced by $\{v_{k+1},\dots,v_n\}\setminus asc_G(v_k)$. An example of such construction is given in Figure \ref{fig10}.
 \medskip

\begin{figure}[h!]
\centering
 \begin{tikzpicture}[edge from parent/.style={draw,-latex}]
    \tikzstyle{level 1}=[sibling distance=40mm]
    \tikzstyle{level 2}=[sibling distance=15mm]
    \tikzstyle{every node}=[circle,draw]

    \node (A) {$v_8$}
        child { node (B) {$v_5$}
                child { node (C) {$v_4$} }
                child { node (D) {$v_2$} }
              }
        child {
            node (E) {$v_7$}
            child { node (F) {$v_6$} }
            child { node (G) {$v_3$} }
            child { node (H) {$v_1$} }
        };

   \node(7) at (8,-1.5) {$v_7$};
   \node(6) at (6.5,-3) {$v_6$} ;
   \draw[-latex] (7) -- (6);

   \begin{scope}[nodes = {left = 8pt},nodes = {draw = none}]
       \node at (A) {$3$};
       \node at (B) {$1$};
       \node at (C) {$2$};
       \node at (D) {$1$};
       \node at (E) {$2$};
       \node at (F) {$3$};
       \node at (G) {$2$};
       \node at (H) {$1$};

       \node at (7) {$2$};
       \node at (6) {$3$};

      \node [below of=F] {A tree $T$ with sorted vertices};

      \node [below of=6] {~~\hspace{2cm}~~The forest $T_4$};
   \end{scope}
\end{tikzpicture}
\caption{An example of construction of the forest $T_k$ for $k=4$.}
\label{fig10}
\end{figure}
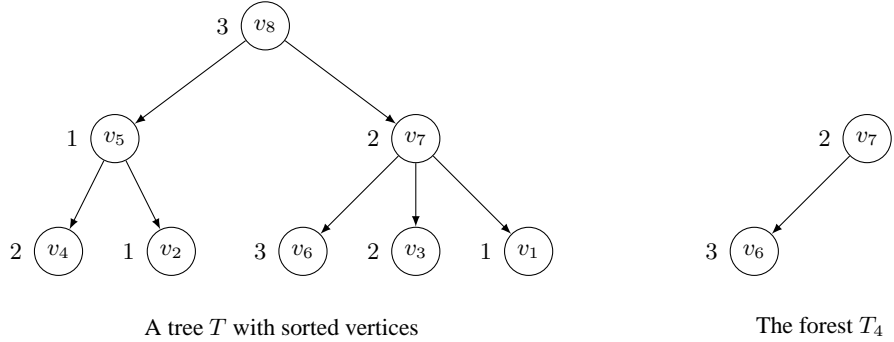

Let $\mathcal{S}$ be the set of feasible solutions of \textsc{Maximal SSG} and $\mathcal{S}_k\subseteq\mathcal{S}$ be the subset satisfying:
\begin{equation}\label{prop_lazy}
  \forall S\in \mathcal{S}_k,\qquad \{v_1,\dots,v_{k-1}\}\subseteq S \text{ and } v_k\notin S
\end{equation}
for $k\in \{1,\dots,n\}$ where $\mathcal{S}_1$ is reduced to subsets $S$ such that $v_1\notin S$. It is obvious that $\{\mathcal{S}_1,\dots,\mathcal{S}_n\}$ is a partition of $\mathcal{S}$ where some parts may be empty.
 \smallskip

Let $k\geq 1$ and let $I_k=(T_k,w,B_k)$ be an instance of \textsc{SSG} where $T_k$ is the forest defined above and $B_k=B-\sum_{i=1}^{k-1}w(v_i)$ where $B_1=B$. Then we show there exists $S\in \mathcal{S}_k$ if and only if $R(T_k,b)=True$ for some
$b\in (B_{k+1},B_k]$ where $R(T_k,b)$ is the boolean table defined in the proof of Proposition \ref{Theo:DynProg_Tree}, i.e. the instance $I_k$ admits a feasible solution for SSG with weight $b\in (B_{k+1},B_k]$.

\smallskip

First, assume there exists $S\in \mathcal{S}_k$. Then

   \begin{itemize}
    \item the restriction $S_k$ of $S$ to $T_k$ satisfies \eqref{dc};
    \item using \eqref{prop_lazy}, we know that $v_k\notin S$, so using Property \ref{property1DAG}, $S\setminus \{v_1,\dots,v_{k-1}\}=S_k$. Using $w(S)\leq B$, we conclude that
    $w(S_k)\leq B_k$;
    \item since $S\in \mathcal{S}_k$ and $T$ is a tree (so, a DAG), it follows that $S$ satisfies \eqref{busy}. But $v_k\notin S$ and $S\cup \{v_k\}$ satisfies \eqref{dc}, so it must be $w(S)+w(v_k)>B$. Thus, $w(S_k)>B_k-w(v_k)=B_{k+1}$.
   \end{itemize}
Hence, we may conclude that $R(T_k,w(S_k))=True$ and $B_{k+1}<w(S_k)\leq B_k$.
\smallskip

Conversely, assume that there exists $b\in (B_{k+1},B_k]$ such that $R(T_k,b)=True$. Then there exists $S_k\subseteq T_k$ such that $S_k$ satisfies \eqref{dc} and $w(S_k)=b\leq B_k$. In the initial graph $T$, we have

\begin{itemize}
 \item $S=S_k\cup \{v_1,\dots,v_{k-1}\}$ satisfies \eqref{dc},
 \item $w(S)=w(S_k)+\sum_{i=1}^{k-1}w(v_i)\leq B$,
 \item Let us prove that $S$ satisfies the maximality constraints \eqref{busy}. Using Remark \ref{rem:lem_local}, let $v$ be a sink of $G-S$; so $v=v_i$ for $i\geq k$.
 We have $w(S)+w(v)\geq w(S)+w(v_k)=b+\sum_{i=1}^{k-1}w(v_i)+w(v_k)> B_{k+1}+\sum_{i=1}^{k-1}w(v_i)+w(v_k)=B_k+\sum_{i=1}^{k-1}w(v_i)=B$.
\end{itemize}
We conclude that $S=S_k\cup \{v_1,\dots,v_{k-1}\} \in \mathcal{S}_k$.
 \bigskip

Then, the value of the optimal solution of \textsc{Maximal SSG} is $\min_{k\leq n}b_k$ with $b_k=\min\{b\in (B_{k+1},B_k]$ for which $R(T_k,b)=True\}$.\hfill\qed
\end{proof}

\begin{proposition}\label{dynprog_strong}
\textsc{SSGW} can be solved using dynamic programming in in-rooted and out-rooted trees.
\end{proposition}

\begin{proof}
 In in-rooted trees, the weak digraph constraints \eqref{sdc} are equivalent to the digraph constraints \eqref{dc}. Hence, \textsc{SSGW} is equivalent to solve \textsc{SSG} and the result holds by Proposition \ref{Theo:DynProg_Tree}.
 \smallskip

 Let us now consider the case of out-rooted trees. Let $P(T,b)$ be a boolean defined on an  out-rooted tree and an integer $b\in \{0,...,B\}$ such that $P(T,b)=True$ if and only if there exists $S\subseteq V$ satisfying the weak digraph constraints \eqref{sdc} with $w(S)=b$. Let $r(T)$ be the anti-root of $T$. Then it is easy to see that $P(T,b)$ is defined recursively as follows.

$$
P(T,b) = \left\{
    \begin{array}{ll}
        False, & \mbox{if } b>w(V), \\
        True, & \mbox{if } b=w(V) \wedge \forall i\in ch(r(T)): P(T_i,w(V_i))=True,\\
        True, &  \mbox{if } \exists I\subseteq ch(r(T)): \forall i\in I: P(T_i,b_i)=True \text{ for some } b_i\in \{0,...,b\}\\
        &  \text{ such that }  \sum_{i\in I} b_i=b \vee \sum_{i\in I} b_i=b-w(r(T)),\\
        False, & \mbox{else.}
    \end{array}
\right.
$$

 \hfill\qed
\end{proof}

The generalization of Proposition \ref{dynprog_strong} to any directed tree is an open problem.
\smallskip

Note that the dynamic programming algorithms presented in Propositions \ref{Theo:DynProg_Tree}, \ref{dynprog_lazy} and \ref{dynprog_strong} can be used to deduce fully polynomial time approximation schemes.

\section{Approximation schemes for \textsc{SSG} and \textsc{Maximal SSG} in DAG}\label{sec:approx}


An instance is a graph $G=(V,A)$, a bound $B$  and a nonnegative weight $w(v)$ for each node $v \in V$.
Using Lemma \ref{lem},  we can suppose that $G$ is a connected DAG. Being a DAG is a hereditary property, so every nonempty subgraph of $G$ possesses a source and a sink. We propose two polynomial approximation schemes for \textsc{SSG} and \textsc{Maximal SSG}, respectively.
They both consist in building a partial solution with an exhaustive search of $k$ nodes ($k$ is a part of the input) and complete it in a greedy manner. In both cases, the time complexity of the algorithm is dominated by the first phase which requires $O(|V|^k|A|)$ elementary operations.

\subsection{{\sc SSG} in DAG}

\begin{algorithm}
\caption{}
\label{algogreed3}
\KwData{a DAG $G=(V,A)$, $B$, $w$ and $k$}
\KwResult{a set of nodes $Sol^*$ satisfying (\ref{dc}) and (\ref{bd})}

$Sol^* \gets \emptyset$

\For{{\bf all} $S\subseteq V$ such that $|S|\leq k$}{


\If{$w(desc_G(S)) \leq B$}{

$V' \gets V\setminus \left( asc_G(\kappa(S)) \cup desc_G(S)\right) $

$G'\gets G[V']$

$Sol_S \gets desc_G(S)$

\While{$V' \neq \emptyset$}{

Choose a source node $z$ of $G',$ maximizing $w(desc_{G'}(z))$

\If{$w(Sol_S)+w(desc_{G'}(z))\leq B$}{

$Sol_S \gets Sol_S \cup desc_{G'}(z)$

$V' \gets V' \setminus desc_{G'}(z)$

}
\Else{
$V' \gets V' \setminus \{z\}$
}
$G'\gets G[V']$

}

\If{$w(Sol_S) > w(Sol^*)$}{

$Sol^*\gets Sol_S$
}

}
}

\Return  $Sol^*$
\end{algorithm}

Algorithm \ref{algogreed3} consists in  building every possible subset $S$ of $V$ such that $|S| \le k$. The nodes
of $desc_G(S)$ are put in a partial solution $Sol_S$ if $w(desc_G(S)) \le B$. Next, the nodes of
$V\setminus \left( asc_G(\kappa(S)) \cup desc_G(S)\right)$ are considered by nonincreasing marginal contribution
(that is $w(desc_G(z)\setminus Sol_S)=w(desc_{G'}(z))$) and $desc_{G'}(z)$ is added to $Sol_S$ if the budget $B$ is not exceeded, until a feasible solution is obtained. The algorithm finally ouputs the best solution $Sol^*$ that was constructed (the one of maximum weight).

\begin{theorem} \label{PTAS-SSG} Algorithm \ref{algogreed3} is a PTAS for  {\sc SSG}  in DAG.
\end{theorem}

\begin{proof}

Let $O$ be an optimal solution to {\sc SSG}. Let $q$ denote $|\kappa(O)|$
and suppose $\kappa(O)=\{v_1, v_2, \ldots, v_q\}$. If $k \ge q$, then the algorithm finds
$\kappa(O)$ and deduces $O$. Henceforth, $k+1 \le q$.
Wlog., assume that the nodes of $\kappa(O)$ have been sorted according to their marginal contribution; so,
let $\hat w(v_{i})$ denote $w(desc_G(v_i) \setminus desc_G(\{v_1, \dots, v_{i-1}\}))$ and $\hat w(v_i) \ge \hat w(v_{i+1})$, $\forall i \in [1..q-1]$.
The value of the optimal solution $w(O)$ is equal to $\sum_{i=1}^q \hat w(v_{i})$. Let $i^*$ be such that
$\hat w(v_{i}) \le \frac{B}{k+1}$ iff $i \ge i^*$.

If $i^* > k+1$, then $\hat w(v_{1}) \ge \cdots \ge \hat w(v_{k+1}) >\frac{B}{k+1}$
and $w(O) >B$, contradiction with (\ref{bd}). Thus, $i^* \le k+1 \le q$ and the algorithm can guess
$\{v_{1}, \dots, v_{i^*-1}\}$ during its first phase. Henceforth, we analyze the iteration of the algorithm where
$S=\{v_{1}, \dots, v_{i^*-1}\}$. Of course, $w(Sol^*)\ge w(Sol_S)$.

During the second phase, if the algorithm inserts $\{ v_{i^*}, v_{i^*+1}, \ldots, v_{q} \}$, then $Sol_S$ is clearly optimal.
Otherwise, let $j$ be the smallest element of $[i^*..q]$ such that $v_{j} \notin Sol_S$.
We cannot add $v_{j}$ to $Sol_S$ because
$w(Sol_S) > B - w(desc_{G'}(v_j))=B - w(desc_G(v_j) \setminus desc_G(Sol_S))$. Since
$w(desc_G(v_j) \setminus desc_G(Sol_S)) \le w(desc_G(v_j) \setminus desc_G(\{v_{1}, \ldots, v_{j-1}\})) = \hat w(v_j)$, we get that
$w(Sol_S) > B - \frac{B}{k+1} \ge \frac{k}{k+1} w(O)$.  \hfill\qed
\end{proof}

\subsection{{\sc Maximal SSG} in DAG}

By Lemma \ref{lem_local}, we know that constraint (\ref{busy1}) can be replaced by (\ref{busy}) in DAG.
Algorithm \ref{algogreed4} consists in building every possible set $S \subseteq V$ such that $|S| \le k$. The nodes
of $desc_G(S)$ are put in a partial solution $Sol_S$ if $w(desc_G(S)) \le B$. Next, the sinks of $G[V \setminus \left(Sol_S \right)]$ are considered
by nondecreasing weight and  greedily added to $Sol_S$ if the budget is not exceeded, until a feasible solution is obtained.
The algorithm finally ouputs the best solution $Sol^*$ that was constructed (the one of minimum weight).

\begin{algorithm}
\caption{}
\label{algogreed4}
\KwData{a DAG $G=(V,A)$, $B$, $w$ and $k$}
\KwResult{a set of nodes $Sol^*$ satisfying (\ref{dc}), (\ref{bd}) and (\ref{busy})}

$Sol^* \gets V$

\For{{\bf all} $S\subseteq V$ such that $|S|\leq k$}{


\If{$w(desc_G(S)) \leq B$}{

$V' \gets V\setminus \left(desc_G(S)\right) $

$G'\gets G[V']$

$Sol_S \gets desc_G(S)$

\While{$Sol_S$ does not satisfy \eqref{busy}}{

Within the sinks of $G'$, choose one, say $z$, of minimum weight


$Sol_S \gets Sol_S \cup \{z\}$

$V' \gets V' \setminus \{z\}$

$G'\gets G[V']$

}

\If{$w(Sol_S) < w(Sol^*)$}{

$Sol^* \gets Sol_S$
}

}
}

\Return  $Sol^*$
\end{algorithm}

\begin{theorem} \label{Theo:PTAS-approxLazy SSG} Algorithm \ref{algogreed4} is a PTAS for  {\sc Maximal SSG}  in DAG.
\end{theorem}

\begin{proof}
Let $O$ be an optimal solution  to {\sc Maximal SSG}. Let $q$ denote $|\kappa(O)|$
and suppose $\kappa(O)=\{v_1, v_2, \ldots, v_q\}$. If $k \ge q$, then the algorithm finds
$\kappa(O)$ and deduces $O$. Henceforth, $k+1 \le q$.
As previously, assume that the nodes of $\kappa(O)$ have been sorted according to their marginal contribution.
Let $\hat w(v_{i})$ denote $w(desc_G(v_i) \setminus desc_G(\{v_1, \ldots, v_{i-1}\}))$, and suppose that $\hat w(v_i) \ge \hat w(v_{i+1})$, $\forall i \in \{1, \ldots,q-1\}$.
The value of the optimal solution $w(O)$ is equal to $\sum_{i=1}^q \hat w(v_{i})$. Observe that for every $j \in [1..q-1]$ and every node
$x \in O \setminus desc_G(\{v_1, \ldots, v_j\})$, it holds that:
\begin{equation} \label{ear}
w(x) \le \hat w(v_{j+1}) \le w(desc_G(\{v_1, \ldots, v_j\}))/j\le w(O)/j
\end{equation}
Consider the iteration of the algorithm where $S=\{v_1, \dots,v_k\}$. The corresponding solution $Sol_S$ consists of
$desc_G(\{v_1, \dots,v_k\})$ plus  some other nodes that are subsequently added in a greedy manner.
Let $z_i$ denote the $i$-th node inserted during the greedy phase. Let $s$ be smallest index such that $z_s \notin O$ (if $z_s$ does not exist, then $Sol_S$ must be optimal).
\\Note that $O \setminus \left( desc_G(\{v_1, \ldots,v_k\}) \cup \{z_1, \ldots ,z_{s-1}\}\right) \neq \emptyset$,
otherwise $z_s$ can be added to $O$, violating (\ref{busy1}).
Let $u$ be a sink of $G[V\setminus \left( desc_G(\{v_1, \ldots,v_k\}) \cup \{z_1, \ldots ,z_{s-1}\} \right)]$, such that $u \in O$. Such a vertex exists because
$G[O]$ is a DAG and $O \supset \left( \{v_1, \ldots,v_k\}) \cup \{z_1, \ldots ,z_{s-1}\} \right)$. Using $(\ref{ear})$, we know that
$w(u) \le w(O)/k$.
The greedy phase of the algorithm consists in adding to the current solution a sink of minimum weight, so $w(u) \ge w(z_s)$. Because $O$ is feasible and $O \supset \left( \{v_1, \ldots,v_k\}) \cup \{z_1, \ldots ,z_{s-1}\} \right)$,
$O\cup \{z_s\}$ must violate the budget constraint, i.e. $w(O) +  w(z_s) >B$.  $Sol_S$
satisfies the budget constraint. We deduce that
$w(Sol^*) \le w(Sol_S) \le B < w(O) +  w(z_s) \le w(O) +  w(u) \le \frac{k+1}{k}w(O)$. \hfill\qed
\end{proof}

\begin{remark}\label{rem:greedy}
Note that for $k=0$, Algorithm \ref{algogreed3} and \ref{algogreed4} are greedy algorithms and it is not difficult to prove that their exact
aproximation bounds are $1/2$ and $2$ for
{\sc SSG} and {\sc Maximal SSG}, respectively.
\end{remark}

\section{Conclusion and perspectives}\label{discuss}
We presented in this article some complexity results for the problem of {\sc (Maximal) subset Sum with (weak) digraph constraints}. We designed complexity results according to the class of the input digraph, namely regular graphs (for SSG), directed acyclic graphs and oriented trees.  It would be interesting to see the tightness of the complexity results in these classes. This was done only for SSG in regular graphs. 
\smallskip


 \bibliographystyle{abbrv}
 \bibliography{bibSubsetSum}

\end{document}